\documentclass[conference]{IEEEtran}
\usepackage{url}
\usepackage{graphicx}
\usepackage{balance}  

\usepackage[cmex10]{amsmath}
\usepackage{amssymb}

\usepackage{main_macro} 

\begin{document}


\title{{\ttlit ModelHub}: Lifecycle Management for Deep Learning}
\title{Towards a Unified Framework for Data and Lifecycle Management for Deep Learning}
\title{Unified Lifecycle Management for Deep Learning}
\title{Towards Unified Data and Lifecycle Management for Deep Learning}



\eat{
%
%
%
%

\numberofauthors{1} 
\author{
%
%
\alignauthor
Hui Miao,~ Ang Li,~ Larry S. Davis,~ Amol Deshpande\\
       \affaddr{University of Maryland, College Park, MD, USA}\\
       \email{{\large{\{}}hui,~angli,~lsd,~amol{\large{\}}}@cs.umd.edu} 
}
}



\author{
    \IEEEauthorblockN{Hui Miao,~ Ang Li,~ Larry S. Davis,~ Amol Deshpande}
    \IEEEauthorblockA{Department of Computer Science \\
    University of Maryland, College Park, MD 20742\\
    {\large{\{}}hui,~angli,~lsd,~amol{\large{\}}}@cs.umd.edu}
}

\maketitle


\begin{abstract}
Deep learning has improved state-of-the-art results in many important fields, and has been the subject of much 
research in recent years, leading to the development of several systems for facilitating deep learning. 
Current systems, however, mainly focus on model building and training phases, while the issues of data management,
model sharing, and lifecycle management are largely ignored.
Deep learning modeling lifecycle generates a rich set of data artifacts, such as learned parameters and training logs, and comprises of several frequently conducted tasks, e.g., to understand the model behaviors and to try out new models. Dealing with such artifacts and tasks is cumbersome and largely left to the users. 
This paper describes our vision and implementation of a data and lifecycle management system for deep learning. First, we generalize model exploration and model enumeration queries from commonly conducted tasks by deep learning modelers, and propose a {\em high-level domain specific language (DSL)}, inspired by SQL, to raise the abstraction level and accelerate the modeling process. 
\icdeadjust{To help a modeler understand models better, we also propose two novel {\em model-comparison schemes} and related algorithms. Second, t}{dropping comparison}
To manage the data artifacts, especially the large amount of checkpointed float parameters, we design a novel model versioning system (\dlv), and a read-optimized parameter archival storage system (\weightstore) that minimizes storage footprint and accelerates query workloads without losing accuracy. \weightstore\ archives versioned models using deltas in a multi-resolution fashion by separately storing the less significant bits, and features a novel progressive query (inference) evaluation algorithm. Third, we show that archiving versioned models using deltas poses a new dataset versioning problem and we develop efficient algorithms for solving it. We conduct extensive experiments over several real datasets from computer vision domain to show the efficiency of the proposed techniques. 
\end{abstract}

\section{Introduction}
\label{sec:intro}

\eat{
\huicomment{a new introduction 2 pages, including abstract}

\begin{list}{$\bullet$}{\leftmargin 0.15in \topsep 0pt \itemsep -3pt}
\item Deep learning is an important method in analytics. 
\item Deep learning anatomy: a DAG structured mapping function. node: templated\_neuron, edge: connection. a DNN program specifies the DAG, input data, and a set of hyper-parameters.
\item Deep learning is end-to-end learning, features are learned directly from data, leading to unique human-in-the-loop life-cycle. 
\item Lifecycle properties: 
\begin{enumerate}
  \item Analysis (describe + compare) the trained models.
  \item Heuristic-driven enumerations of the models. 
  \item A rich set of artifacts with lineages.
\end{enumerate}
\item Lifecycle disadvantages under current systems: 
\begin{enumerate}
  \item repetition and time consuming of human effort of doing description and enumeration via imperative programs;
  \item heavy turnover and footprint: long training time, less usage frequency of older models; a trade-off the user needs to decide: save more models for reuse or delete + retrain.
  \item redundancy in large models; similar models exist; may not need all models with the same frequency and accuracy. 
  \item complicated and incoherent setups: hard to share, reuse, reproduce others' models. 
\end{enumerate}
\item So we propose \modelhub\ system to tackle the disadvantages:
\begin{enumerate}
  \item Increase abstraction level and lead to optimization: 
  \begin{enumerate}
    \item we propose a VCS based system. use friendly query templates to capture the process in modeling lifecycle.
    \item we propose a DQL language, to help the user enumerate models, and encode model selection criteria. 
  \end{enumerate}
  \item Treat weight matrices as first class data type:
    \begin{enumerate}
      \item we identify a matrix comparison operator: alignment based matrix reordering, which is useful for comparison queries and model storages.
      \item we explore the problem of storing many models which have both highly structured information and high entropy float matrices with low precision tolerance. 
      \item we develop model query directly on compressed storages, with guarantee of no errors, and at the same time faster query processing.
      It eases the trade-off between model storage and usage. we can compress all models with less storages, and do pay-as-you-go style of uncompressing to ensure no errors.
    \end{enumerate}
  \item the technique proposed extends naturally to a collaborative format and environment which has rich lineages and enables sharing, reusing, reproducing models. 
\end{enumerate}
\item Contributions:
\begin{enumerate}
  \item we are the first addressing model lifecycle management issues in deep learning. 
  \item we are the first proposing full-lifecycle declarative constructs for deep learning modeling and show their implementations.
  \item we treat deep learning weight matrix as first class data types, we formulate the storage problem of a large set of float parameter matrices with deep learning lifecycle specific query workloads, and explore the low precision tolerance and accuracy trade-off to save the data type.
  \item by using bytewise compression, we present a storage framework and model evaluation query acceleration technique with guarantees of no errors.
  \item on the weight matrix data type, we identify a new problem of alignment based matrix reordering problem, appeared in deep learning model comparison and storage. We shown the problem is NP-hard and propose efficient greedy algorithms.
\end{enumerate}
\item Our results show the techniques we propose are useful in real life models, and performance well on synthetic models. 
\item Outline: Sec2: Preliminary, Sec3: System Overview, Sec4: Optimizations for Weight Matrices, Sec. 4.1 Alignment Operator. Sec 4.2: Storage. Sec 5. Exp. Sec 6. Related work. Sec 7. Conclusion.
\end{list}

\quad \vfill
\pagebreak 
}

Deep learning models, also called {\em deep neural networks} (DNN), have dramatically improved the state-of-the-art results for many important reasoning and learning tasks including speech recognition, object recognition, and natural language processing in recent years \cite{lecun2015nature}. 
Learned using massive amounts of training data, \dnn\ models have superior generalization capabilities, and the intermediate layers in many deep learning models have been proven useful in providing effective semantic features that can be used with other learning techniques and are applicable to other problems.
However, there are many critical large-scale data management issues in learning, storing, sharing, and using deep learning models, which are largely ignored by researchers today, but are coming to the forefront with the increased use of deep learning in a variety of domains. 
In this paper, we discuss some of those challenges in the context of the modeling lifecycle, and propose a comprehensive system to address them. Given the large scale of data involved (both training data and the learned models themselves) and the increasing need for high-level declarative abstractions, we argue that database researchers should play a much larger role in this area. Although this paper primarily focuses on deep neural networks, similar data management challenges are seen in lifecycle management of others types of ML models like logistic regression, matrix factorization, etc.


\topic{DNN Modeling Lifecycle and Challenges:}
Compared with the traditional approach of {\em feature engineering} followed by {\em model training}~\cite{ce2014sigmod}, deep learning is an end-to-end learning approach, i.e., the features are not given by a human but are learned in an automatic manner from the input data. Moreover, the features are complex and have a hierarchy along with the network representation. 
This requires less domain expertise and experience from the modeler, but
understanding and explaining the learned models is difficult; why even well-studied models work so well is still a mystery\eat{ is unknown in theory} and under active research. 
Thus, when developing new models, changing the learned model (especially its network structure and {\em hyper-parameters}) becomes an empirical search task. 


\begin{figure}[!t]
\centering
\includegraphics[width=1.0\linewidth]{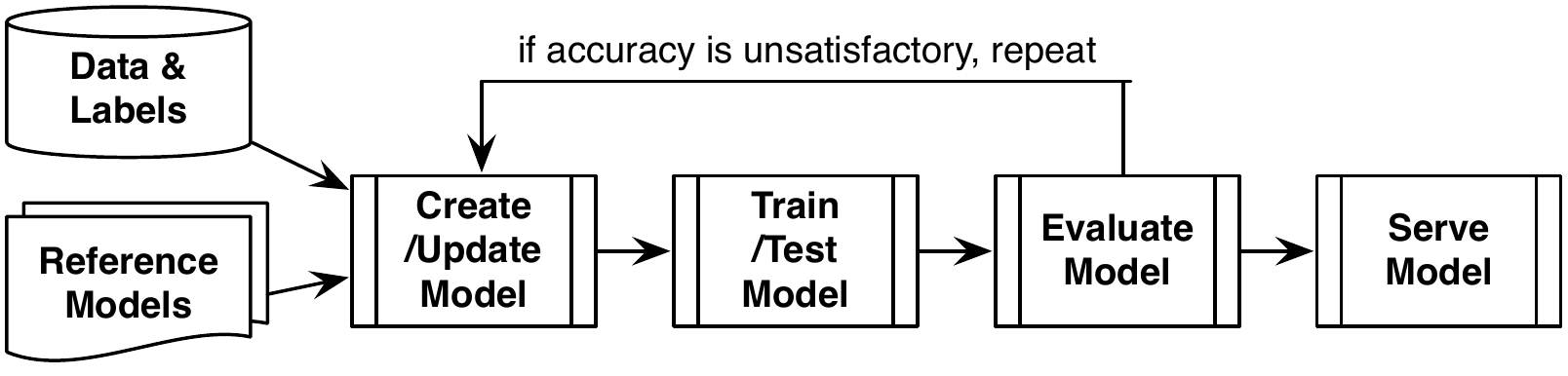}
\caption{Deep Learning Modeling Lifecycle}
\label{fig:lifecycle}
\end{figure}

In Fig.~\ref{fig:lifecycle}, we show a typical deep learning modeling lifecycle (we present an overview of deep neural networks in the next section). Given a prediction task, a modeler often starts from well-known models that have been successful in similar task domains; she then specifies input training data and output loss functions, and repeatedly adjusts the \dnn\ on operators
and connections like Lego bricks, tunes model hyper-parameters, trains and evaluates the model, and repeats this loop until prediction accuracy does not improve. Due to a lack of understanding about why models work, the adjustments and tuning inside the loop are driven by heuristics, e.g.,
adjusting hyper-parameters that appear to have a significant impact on the learned weights, applying novel layers or tricks seen in recent empirical studies, and so on. Thus, many similar models are trained and compared, and a series of model variants needs to be explored and developed. 
Due to the expensive learning/training phase, each iteration of the modeling loop takes a long period of time 
and produces many (checkpointed) snapshots of the model. As we noted above, this is a common workflow across many other ML models as well.

Current systems (Caffe~\cite{caffe2014mm}, Theano, Torch, TensorFlow~\cite{tensorflow}, etc.) mainly focus on model building and training phases, while the issues of data management, model sharing, and lifecycle management are largely ignored. Modelers are required to write external imperative scripts, edit configurations by hand and manually maintain a manifest of model variations that have been tried out; not only are these tasks irrelevant to the modeling objective, but they are also challenging and nontrivial due to the complexity of the model as well as large footprints of the learned models. More specifically, the tasks and data artifacts in the modeling lifecycle expose several systems and data management challenges, which include:
\begin{list}{$\bullet$}{\leftmargin 0.10in \topsep -2pt} 
       \item \textbf{Understanding \& Comparing Models}: It is difficult to keep track of the many models developed and/or understand the differences amongst them. Differences among both the metadata about the model (training sample, hyperparameters, network structure, etc.), as well as the actual learned parameters, are of interest. It is common to see a modeler write all\eat{ model} configurations in a\eat{n experiment} spreadsheet to keep track of temporary folders of input\eat{ data}, setup scripts, snapshots and logs, which is not only a cumbersome but also an error-prone process. \eat{Though it is fine to view the measurements, understanding the model difference is less-principled and requires a lot more work.
       repetition and time consuming sub-steps in model enumerations}
       \item \textbf{Repetitive Adjusting of Models}: The development lifecycle itself has time-consuming repetitive sub-steps, such as adding a layer at different places to adjust a model, searching through a set of hyper-parameters for the different variations, reusing learned weights to train models, etc., which currently have to be performed manually.
       \item \textbf{Model Versioning}: Similar models are possibly trained and run multiple times, reusing others' weights as initialization, either because of a changed input or discovery of an error. 
       There is thus a need to keep track of multiple model versions and their relationships over time, although the utilities of different models are very different. 
       \item \textbf{Parameter Archiving}: The storage footprint of deep learning models tends to be very large. Recent top-ranked models in the ImageNet task have billions of floating-point parameters and require hundreds of MBs to store one snapshot during training.
       Due to resource constraints, the modeler has to limit the number of snapshots, even drop all snapshots of a model at the cost of retraining when needed. 
      \item \textbf{Reasoning about Model Results:} Another key data artifact that often needs to be reasoned about is the results of running a learned model on the training or testing dataset. By comparing the results across different models, a modeler can get insights into difficult training examples or understand correlations between specific adjustments and the performance.
\end{list}
In addition, although not a focus of this paper, sharing and reusing models is not easy, especially because of the large model sizes and specialized tools used for learning and modeler generated scripts in the lifecycle.

\topic{ModelHub:}
In this paper, we propose the \modelhub\ system to address these challenges. 
The \modelhub\ system is not meant to replace popular training-focused \dnn\ systems, but rather designed to be used with them to accelerate modeling tasks and manage the rich set of lifecycle artifacts.
It consists of three key components: (a) a model versioning system (\DLV) to store, query and aid in understanding the models and their versions, (b) a model network adjustment and hyper-parameter tuning domain specific language (\DQL) to serve as an abstraction layer to help modelers focus on the creation of the models\eat{ instead of repetitive steps in the lifecycle}, (c) a hosted deep learning model sharing system (\modelhub) to exchange \DLV\ repositories and enable publishing, discovering and reusing models from others.

The key features and innovative design highlights of \modelhub\ are: 
\textbf{(a)} We use a \cmd{git}-like VCS interface as a familiar user interface to let the modeler manage and explore the created models in a repository, and an \cmd{SQL}-like model enumeration DSL to aid modelers in making and examining multiple model adjustments easily. 
\icdeadjust{
\textbf{(b)} Because model comparison is less-principled today, we propose two new model understanding and comparison schemes. }{no comparison}
\textbf{(b)} Behind the declarative constructs, \modelhub\ manages different artifacts in a split back-end storage: structured data, such as network structure, training logs of a model, lineages of different model versions, output results, are stored in {\em a relational database}, while learned float-point parameters of a model are viewed as a set of float matrices and managed in {\em a read-optimized archival storage (\weightstore)}. 
\textbf{(c)} Parameters dominate the storage footprint and floats are well-known at being difficult to compress. We study \weightstore\ implementation thoroughly under the context of \dnn\ query workload and advocate a segmented approach  to store the learned parameters, where the low-order bytes are stored independently of the high-order bytes. We also develop novel model evaluation schemes to use high order bytes solely and progressively uncompress less-significant chunks if needed to ensure the correctness of an inference query. 
\textbf{(d)} Due to the different utility of developed models, archiving versioned models using parameter matrix deltas exhibits a new type of dataset versioning problem which not only optimizes between storage and access tradeoff but also has model-level constraints.
\textbf{(e)} Finally, the VCS model repository design extends naturally to a collaborative format and online system which contain rich model lineages and enables sharing, reusing, reproducing \dnn\ models which are compatible across training systems.

\topic{Contributions}: 
Our key research contributions are:
\begin{list}{$\bullet$}{\leftmargin 0.10in \topsep -2pt} 
  \item We propose the first comprehensive \dnn\ lifecycle management system, study its design requirements, and propose declarative constructs (\DLV\ and \DQL) to provide high-level abstractions.
  \icdeadjust{
  \item We propose two new model understanding and comparison schemes, and study a new matrix reordering problem (for matrix alignment). 
  We analyze its complexity and propose greedy algorithms for solving it.
  }{no comparison}
  \item We develop \weightstore, a read-optimized archival storage system for dealing with a large collection of versioned float matrices. 
 \item We formulate a new dataset versioning problem with {\em co-usage constraints}, analyze its complexity, and design  efficient algorithms for solving it.
 \item We develop a progressive, approximate query evaluation scheme that avoids reading low-order bytes of the parameter matrices unless necessary.
  \item We present a comprehensive evaluation of \modelhub\ that shows the proposed techniques are useful in real life models, and scale well on synthetic models. 
\end{list}

\topic{Outline}: 
In Section~\ref{sec:preliminary}, we provide background on related topics in \dnn\ modeling lifecycle. In Section~\ref{sec:sys_overview}, we present an overview of \modelhub, and discuss the declarative interfaces. We describe the parameter archival store (\weightstore) in Section~\ref{sec:binary_storage}, present an experimental evaluation in Section~\ref{sec:experiments}, and closely related work in Section~\ref{sec:related_work}. 

\eat{
\begin{figure*}[!t]
\subfigure[Anatomy of A \dnn\ Model (LeNet)]{
\raisebox{5mm}{
  \includegraphics[totalheight=0.15 \textheight]{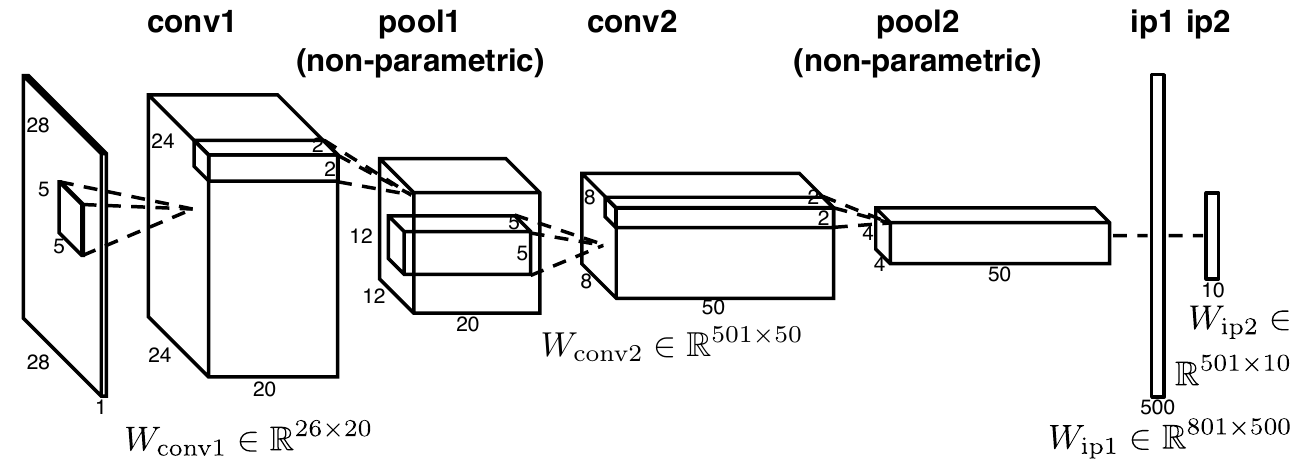}
  \label{fig:onesnapshot}
}}
\subfigure[Relationships of Model Versions and Weight Snapshots]{
  \includegraphics[totalheight=0.2\textheight]{snapshots.pdf} 
  \label{fig:snapshots}
}
\caption{Illustration of Weight Parameter Artifacts in a Modeling Process}
\label{fig:weight_datamodel}
\end{figure*}
}

\begin{figure}[!t]
\centering
\includegraphics[width=\linewidth]{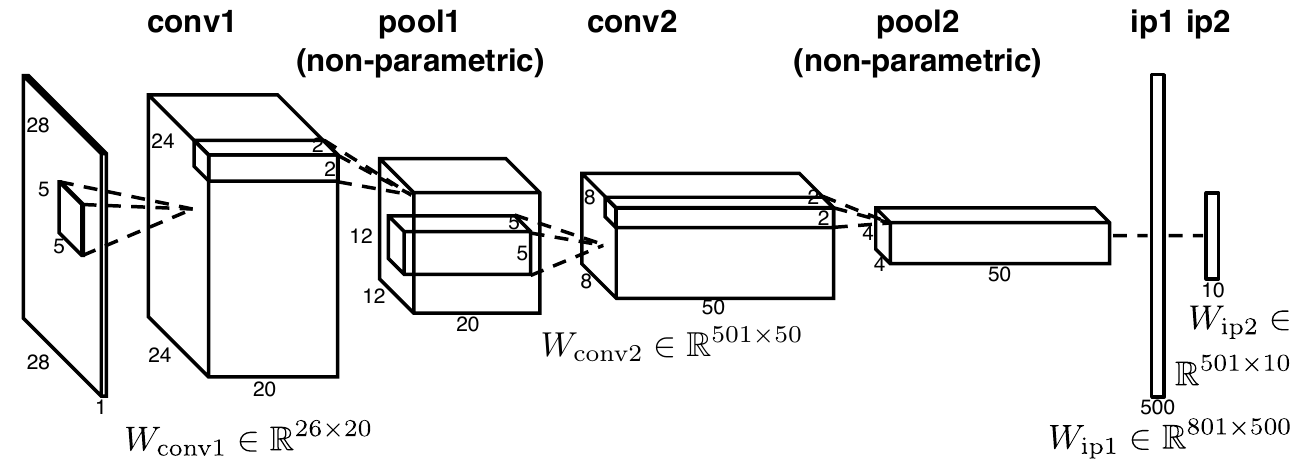}
\caption{Anatomy of A \dnn\ Model (LeNet)}
\label{fig:onesnapshot}
\end{figure}

\section{Background\eat{ on Lifecycle}}
\label{sec:preliminary}

To support our design decisions, we overview the artifacts and common task practices in \dnn\ modeling lifecycle. \eat{We also examine the dataset versioning problem from recent database community research and point out the inefficiencies for \dnn\ lifecycle management.}

\topic{Deep Neural Networks:} 
\eat{We begin with a brief, simplified overview.}
A deep learning model is a deep neural network (\dnn) consisting of many layers having nonlinear activation functions that are capable of representing complex transformations between input data and desired output.
Let $\mathbb D$ denote a data domain and $\mathbb O$ denote a prediction label domain (e.g., $\mathbb D$ may be a set of images; $\mathbb O$ may be the names of the set of objects we wish to recognize, i.e, {\em labels}). As with any prediction model, a DNN is\eat{ essentially} a mapping function $f:\mathbb{D}\to\mathbb{O}$ that
minimizes a certain loss function $L$, and is of the following form:\\
\begin{figure}[!h]
\centering
\includegraphics[width=0.9\linewidth]{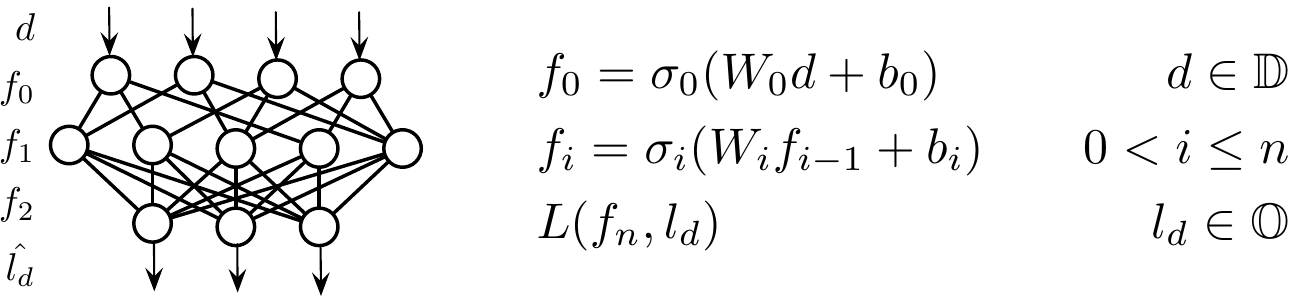}
\end{figure}

\noindent Here $i$ denotes the layer number,  $(W_i, b_i)$ are learnable weights and bias parameters in layer $i$, and $\sigma_i$ is an activation function that non-linearly transforms the result of the previous layer (common activation functions include sigmoid, ReLU, etc.). 
Given a learned model and an input $d$, applying $f_0, f_1, ..., f_{n}$ in order gives us the prediction label for that input data. 
In the training phase, the model parameters are learned by minimizing $L(f_{n}, l_d)$,
typically done through iterative methods, such as {\em stochastic gradient descent}.


Fig.~\ref{fig:onesnapshot} shows a classic {\em convolutional DNN}, LeNet. LeNet is proposed to solve a prediction task from handwritten images to digit labels $\{0\cdots9\}$. In the figure, a cube represents an intermediate tensor, while the dotted lines are unit transformations between tensors. More formally, a layer, $L_i : (W,H,X) \mapsto  Y$, is a function which defines data transformations from tensor $X$ to tensor $Y$. $W$ are the parameters which are learned from the data, and $H$ are the hyperparameters which are given beforehand. A layer is non-parametric if $W = \emptyset$.

In the computer vision community, the layers defining transformations are considered building blocks of a \dnn\ model, and referred to using a conventional name, such as {\em full} layer, {\em convolution} layer, {\em pool} layer, {\em normalization} layer, etc. The chain is often called the network \emph{architecture}. The LeNet architecture has two convolution layers, each followed by a pool layer, and two full layers, shown with layer shapes and hyperparameters in Fig.~\ref{fig:onesnapshot}.
Moreover, winning models in recent ILSVRC (ImageNet Large Scale Vision Recognition Competitions) are shown in Table~\ref{tb:network_comparison}, with their architectures described by a composition of common layers in regular expressions syntax for illustrating the similarities (Note the activation functions and detailed connections are omitted). \eat{As one can see, common layers are building blocks of \dnn\ models.} 

\dnn\ models are learned from massive data based on some architecture, and 
modern successful computer vision DNN architectures consist of a large number of float weight parameters (\emph{flops}) 
\eat{The number of float parameters (\emph{flops}) are} shown in Table~\ref{tb:network_comparison}, resulting in large storage footprints (GBs) and long training times (often weeks). Furthermore, the training process is often checkpointed and variations of models need to be explored, leading to many model copies.

\topic{Modeling Data Artifacts}: 
Unlike many other prediction methods, \dnn\ modeling results in a very large number of weight parameters, a rich set of hyperparameters, and learning measurements, which are used in unique ways in practice, resulting in a mixture of structured data, files and binary floating number artifacts: 
\begin{list}{$\bullet$}{\leftmargin 0.10in \topsep -2pt} 
\item \textbf{Non-convexity \& Hyperparameters}: A DNN model is typically non-convex, and $\{W\}$ is a local optimum of the underlying loss-minimization problem. Optimization procedure employs many tricks to reach a solution quickly~\cite{sgdtrick2012dnntrickbook}. The set of hyperparameters (e.g., learning rate, momentum) w.r.t. to the optimization algorithm need to be maintained. 
\item \textbf{Iterations \& Measurements}: Models are trained iteratively and checkpointed periodically due to the long running times.  A set of learning measurements are collected in various logs, including objective loss values and accuracy scores. 
\item \textbf{Fine-tuning \& Snapshots}: Well-known models are often learned from massive real-world data (ImageNet), and require large amounts of resources to train; when prediction tasks do not vary much (e.g., animal recognition vs dog recognition), the model parameters are reused as initializations and adjusted using new data; this is often referred to as fine-tuning. On the other hand, not all snapshots can be simply deleted, as the convergence is not monotonic.
\item \textbf{Provenance \& Arbitrary Files}: Alternate ways to construct architectures or to set hyperparameters lead to human-in-the-loop model adjustments. Initialization, preprocessing schemes, and hand-crafted scripts are crucial provenance information to explore models and reproduce results.
\end{list}


\begin{table}[!t]
\centering
\begin{tabular}{ccccc} 
\toprule
\textbf{Network} & \textbf{Architecture} (in regular expression) & \textbf{$|W|$ (flops)}
\\ 
\midrule
LeNet \cite{lenet90} & $(L_{conv}L_{pool})\{2\}L_{ip}\{2\}$ & $4.31 \times 10^{5}$\\ 
\midrule
AlexNet \cite{alexnet2012imagenet} & $(L_{conv}L_{pool})\{2\}(L_{conv}\{2\}L_{pool})\{2\}L_{ip}\{3\}$ & $6\times10^7$  \\ 
\midrule
VGG \cite{vgg14} & $(L_{conv}\{2\}L_{pool})\{2\}(L_{conv}\{4\}L_{pool})\{3\}L_{ip}\{3\}$ & $1.96 \times 10^{10}$  \\ 
\midrule
ResNet \cite{resnet} & $(L_{conv}L_{pool})(L_{conv})\{150\}L_{pool}L_{ip}$ & $1.13\times10^{10}$  \\ 
\bottomrule
\vspace{0.1mm}
\end{tabular}
\caption{Popular CNN Models for Object Recognition}
\label{tb:network_comparison}
\end{table}

\topic{Model Adjustment}:
In a modeling lifecycle for a prediction task, the \emph{update-train-evaluate} loop is repeated in daily work, and many model variations are adjusted and trained. 
In general, once data and loss are determined, model adjustment can be done in two orthogonal steps: a) network architecture adjustments where layers are dropped or added and
layer function templates are varied, and b) hyperparameter selections, which affect the behavior of the optimization algorithms. 
There is much work on search strategies to enumerate and explore both. 

\topic{Model Sharing}: 
Due to the good generalizability, long training times, and verbose hyperparameters required for large \dnn\ models, there is a need to share the trained models. 
Jia et al.~\cite{caffe2014mm} built an online venue (Caffe Model Zoo) to share models. Briefly, Model Zoo is part of a github repository\footnote{Caffe Model Zoo: \url{https://github.com/BVLC/caffe/wiki/Model-Zoo}} with a markdown file edited collaboratively. To publish models, modelers add an entry with links to download trained parameters in \cmd{caffe} format. Apart from the caffe community, similar initiatives are in place for other training systems.

\section{ModelHub System Overview}
\label{sec:sys_overview}

We show the \modelhub\ architecture including the key components and their interactions in
Fig.~\ref{fig:sys_arch}. \eat{At a high level, as a modeling lifecycle management tool, first it has to run on local machine and integrate with popular \dnn\ systems, such as \cmd{caffe}, \cmd{torch}, \cmd{tensorflow}; second an online module is served as a cloud host to save and exchange model versions.}
At a high level, the \modelhub\ functionality is divided among a local component and a remote component. The local functionality includes the integration with popular DNN systems such as \cmd{caffe}, \cmd{torch}, \cmd{tensorflow}, etc., on a local machine or a cluster. The remote functionality includes sharing of models, and their versions, among different groups of users. We primarily focus on the local functionality in this paper.

On the local system side, \DLV\ is a version control system (VCS) implemented as a command-line tool
(\dlv), that serves as an interface to interact with the rest of the local and remote components. Use of a specialized
VCS instead of a general-purpose VCS such as \cmd{git} or \cmd{svn} allows us to 
better portray and query the internal structure of the artifacts generated in a modeling lifecycle, such as network
definitions, training logs, binary weights, and relationships between models. The key utilities of
\dlv\ are listed in Table~\ref{tb:dlvcmds}, grouped by their purpose; we explain these in further
detail in
Sec.~\ref{subsec:query}. \DQL\ is a DSL we propose to assist modelers in deriving new models; the
\DQL\ query parser and optimizer components in the figure are used to support this language. The
\emph{model learning module} 
interacts with external deep learning tools
that the modeler uses for training and testing. They are essentially wrappers on specific \dnn\ systems that extract and reproduce modeling artifacts. 
Finally, the \modelhub\ service is a hosted toolkit to support publishing,
discovering and reusing models, and serves similar role for \dnn\ models as \emph{github} for software development or \emph{DataHub} for data science~\cite{datahub}.

\begin{figure}[!t]
\centering
\includegraphics[width=0.75\linewidth]{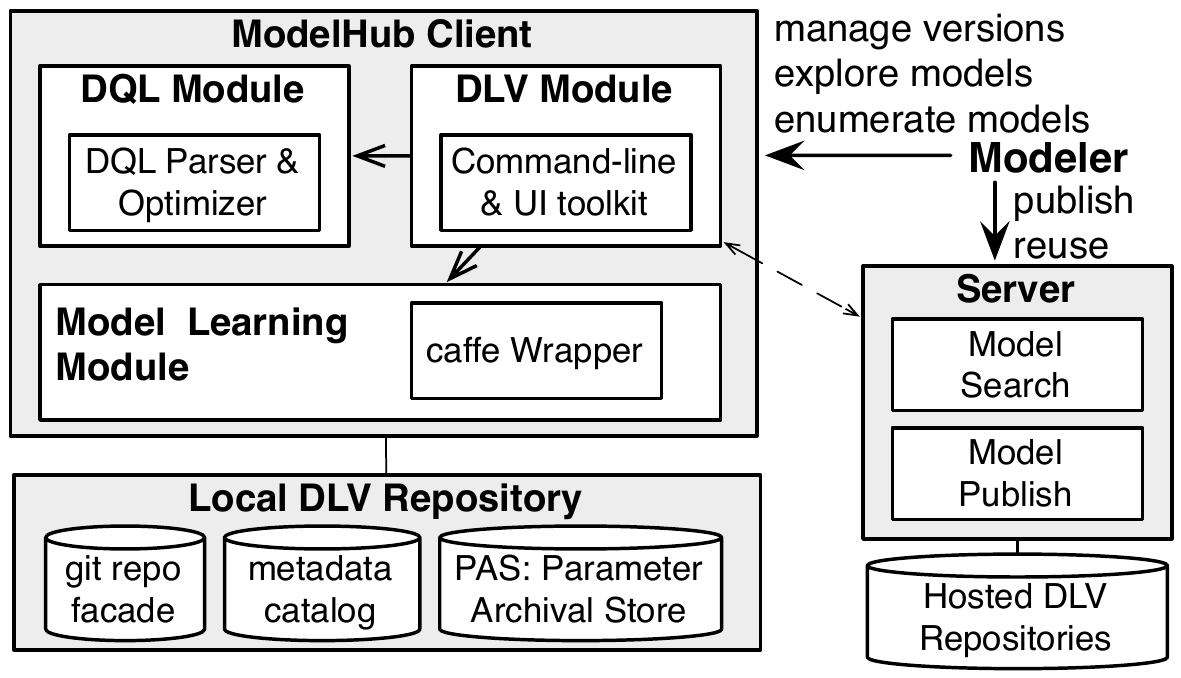}
\caption{\modelhub\ System Architecture}
\label{fig:sys_arch}
\end{figure}

\subsection{Data Model}
\label{subsec:datamodel}
\modelhub\ works with two data models: 
a conceptual \dnn\ model, and a data model for the versions in a \DLV\ repository.

\topic{\dnn\ Model:}
A \dnn\ model can be understood in different ways, as one can tell from the different model creation APIs in popular deep learning systems. In the formulation mentioned in Sec.~\ref{sec:intro}, if we view a function $f_i$ as a node and
dependency relationship $(f_i, f_{i-1})$ as an edge, it becomes a directed acyclic graph (DAG).
Depending on the granularity of the function in the DAG, either at the tensor arithmetic operator level (add, multiply), or at a logical composition of those operators (convolution layer, full layer), it forms different types of DAGs. In \modelhub, we consider a \dnn\ model node as a composition of unit operators (layers), often adopted by computer vision models. The main reason for this decision is that we focus on productivity improvement in the lifecycle, rather than implementation efficiencies of training and testing. 

\eat{
Though a \dnn\ model essentially is a parametric nested mapping function learned from massive examples. However, in practice, a \dnn\ model is represented in different ways.

Popular training systems often use a \emph{graph} construction API at different levels to present the nested function, where the graph \emph{node} is referred as \emph{operator}, \emph{layer}, \emph{gate} and etc. To illustrate the difference, if we view a function $f_i$ (a \emph{layer}) as a node and dependency relationship $(f_i, f_{i-1})$ as an edge, it becomes a directed acyclic graph (DAG) over layers, while if a graph node is defined on $+$, $\cdot$, $\sigma_i$ these basic tensor arithmetic \emph{operators}, its a DAG over operators. 
Modelers think and communicate in a high level (e.g. convolution layer, full layer), while may program in an API dialect of low level arithmetic operators.

In \modelhub, we consider a \dnn\ model node as a composition of unit operators (layers).\eat{, often adopted by computer vision models.} The main reason for the decision is that we focus on the productivity improvement in the lifecycle, rather than the implementation efficiencies for training and testing. 
}

\topic{VCS Data Model:}
When managing \dnn\ models in the VCS repository, a \emph{model version} represents the contents in
a single version. It consists of a network definition, a collection of weights (each of which is a
        value assignment for the weight parameters), a set of extracted metadata (such as
            hyper-parameter, accuracy and loss generated in the training phase), and a collection
        of files used together with the model instance (e.g., scripts, datasets). In addition, we
        enforce that a \emph{model version} must be associated with a human readable name for better
        utility, which reflects the logical groups of a series of improvement efforts over a \dnn\ model in practice.

In the implementation, model versions can be viewed as a relation \emph{model\_version}$($\underline{name, id}, N, W, M,
        F$)$, where id is part of the primary key of model versions and is auto-generated to
distinguish model versions with the same name. In brief, $N, W, M, F$ are the network definition,
            weight values, extracted metadata and associated files respectively. The DAG, \emph{N},
            is stored as two tables: Node$($\underline{id, node}, A$)$, where $A$ is a list of
            attributes such as layer name, and Edge$($\underline{from, to}$)$. $W$ is managed in our learned parameter storage (\weightstore, Sec.~\ref{sec:binary_storage}). $M$, the metadata, captures the provenance information of training and testing a particular model; it is extracted from training logs by the wrapper module, and includes the hyperparameters when training a model, the loss and accuracy measures at some iterations, as well as dynamic parameters in the optimization process, such as learning rate at some iterations. Finally, $F$ is file list marked to be associated with a model version, including data files, scripts, initial configurations, and etc.
            Besides a set of \emph{model version}s, the lineage of the \emph{model versions} are captured using a
            separate \emph{parent}$($\underline{base, derived}, commit$)$ relation. All of these relations are maintained/updated in a relational backend when
            the modeler runs the different \dlv\ commands that update the repository.


\begin{table}[!t]
\centering
\begin{tabular}{p{0.275\linewidth}p{0.11\linewidth}p{0.48\linewidth}}
\toprule
Type & Command & Description\\ 
\midrule
\multirow{3}{\linewidth}{model version management}
                     & \cmd{init}    & Initialize a \dlv\ repository.\\ 
                     & \cmd{add}     & Add model files to be committed.\\ 
                     & \cmd{commit}  & Commit the added files.\\ 
                     & \cmd{copy}    & Scaffold model from an old one.\\
                     & \cmd{archive} & Archive models in the repository. \\
\midrule
\multirow{3}{\linewidth}{model exploration}
                     & \cmd{list}    & List models and related lineages. \\ 
                     & \cmd{desc}    & Describe a particular model. \\ 
                     & \cmd{diff}    & Compare multiple models. \\ 
                     & \cmd{eval}    & Evaluate a model with given data. \\ 
\midrule
model enumeration   & \cmd{query}   & Run \DQL\ clause. \\ 
\midrule
\multirow{3}{\linewidth}{remote interaction}
                     & \cmd{publish}  & Publish a model to ModelHub.\\ 
                     & \cmd{search}  & Search models in ModelHub.\\ 
                     & \cmd{pull}   & Download from ModelHub.\\ 
\bottomrule
\end{tabular}
\label{tb:dlvcmds}
\caption{A list of key \dlv\ utilities. }
\end{table}

\subsection{Query Facilities}
\label{subsec:query}
Once the \dnn\ models and their relationships are managed in \DLV, the modeler can interact with them easily. 
The query facilities we provide can be categorized into two types: a) model exploration queries and b) model enumeration queries. 

\subsubsection{Model Exploration Queries}
Model exploration queries interact with the models in a repository, and are used to
understand a particular model, to query lineages of the models, and to compare several models. 
For usability, we design it as 
query templates via \dlv\ 
sub-command\eat{ with options}, similar to other VCS. 

\topic{List Models \& Related Lineages:} 
By default, the query lists all versions of all models including their commit descriptions and
parent versions; it also takes options, such as showing results for a particular model, or limiting
the number of versions to be listed. 
\\{\small{\verb|   dlv list [--model_name] [--commit_msg] [--last]|}}

\topic{Describe Model:} 
\dlvcmd{desc} shows the extracted metadata from a model version,
    such as the network definition, learnable parameters, execution footprint (memory and runtime),
    activations of convolutional DNNs, weight matrices, and evaluation results across iterations.
    Note the activation is the intermediate output of a \dnn\ model in computer vision and often
    used as an important tool to understand the model. The current output formats are a result of
    discussions with computer vision modelers to deliver tools that fit their needs. In
    addition to printing to console, the query supports {\bf HTML output} for displaying the images and visualizing the weight distribution.
\\{\small{\verb|   dlv desc [--model_name |\texttt{$\vert$}\verb| --version] [--output]|}}

\topic{Compare Models:} 
\dlvcmd{diff} takes a list of model names or version ids and allows the
modeler to compare the \dnn\ models. Most of \cmd{desc} components are aligned and returned in the
query result side by side. 
\icdeadjust{We discuss it in Sec.~\ref{subsec:comparison}.}{drop comparison}
\\{\small{\verb|   dlv diff [--model_names |\texttt{$\vert$}\verb| --versions] [--output]|}}

\topic{Evaluate Model:} 
\dlvcmd{eval} runs test phase of the managed models with an optional config
specifying different data or changes in the current hyper-parameters. The main usages of exploration
query are two-fold: 1) for the users to get familiar with a new model, 2) for the user\eat{ who wants}
to test known models on different data or settings. The query returns the accuracy and optionally the activations. 
It is worth pointing out  that complex evaluations can be done via model enumeration queries in \DQL.
\\{\small{\verb|   dlv eval [--model_name |\texttt{$\vert$}\verb| --versions] [--config]|}}

\subsubsection{Model Enumeration Queries}
\label{subsec:dql}
Model enumeration queries are used to explore variations of currently available models in a
repository by changing network structures or tuning hyper-parameters.
There are several operations that need to be done in order to derive new models: 1) Select
models from the repository to improve; 2) Slice particular models to get reusable components; 3)
Construct new models by mutating the existing ones; 4) Try the new models on different
hyper-parameters and pick good ones to save and work with. When enumerating models, we also want to
stop exploration of bad models early.

To support this rich set of requirements, 
we propose the \DQL\ domain specific language, that can be executed using ``\dlvcmd{query}''. Challenges of designing the language are: a) the data model is a mix of relational and the graph data models and b) the enumeration includes hyper-parameter tuning as well as network structure mutations, which are very different operations. We omit a thorough explanation of the language due to space constraints, and instead show the key operators and constructs\eat{ of the language} along with a set of examples (Query~\ref{dql:select}$\sim$\ref{dql:eval}) to show how requirements are met.

\begin{figure}
\begin{lstlisting}[caption={\DQL\ \cmd{select} query to pick the models.\quad\quad\quad\quad\quad\quad\quad\quad\quad\quad\quad}, label=dql:select]
select m1
where  m1.name like "alexnet_%" and
       m1.creation_time > "2015-11-22" and
       m1["conv[1,3,5]"].next has POOL("MAX")
\end{lstlisting}

\begin{lstlisting}[caption={\DQL\ \cmd{slice} query to get a sub-network.\quad\quad\quad\quad\quad\quad\quad\quad\quad\quad}, label=dql:slice]
slice m2 from m1
where m1.name like "alexnet-origin%"
mutate m2.input = m1["conv1"] and
       m2.output = m1["fc7"]
\end{lstlisting}

\begin{lstlisting}[caption={\DQL\ \cmd{construct} query to derive more models on existing ones.\quad}, label=dql:construct]
construct m2 from m1
where  m1.name like "alexnet-avgv1%" and 
       m1["conv*($1)"].next has POOL("AVG")
mutate m1["conv*($1)"].insert = RELU("relu$1")
\end{lstlisting}

\begin{lstlisting}[caption={\DQL\ \cmd{evaluate} query to enumerate models with different network architectures, search hyper-parameters, and eliminate models.}, label=dql:eval]
evaluate m
from "query3"
with config = "path to config"
vary config.base_lr in [0.1, 0.01, 0.001] and
     config.net["conv*"].lr auto and
     config.input_data in ["path1", "path2"]
keep top(5, m["loss"], 100)
\end{lstlisting}
\end{figure}

\topic{Key Operators:} 
We adopt the standard SQL syntax to interact with the repository. \DQL\ views
the repository as a single model version table. \eat{As mentioned in Sec~\ref{subsec:datamodel}, a}A model
version instance is a DAG, which can be viewed as object types in modern SQL conventions. In \DQL,
        \eat{DAG level} attributes can be referenced using attribute names (e.g. \verb|m1.name|,
                \verb|m1.creation_time|, \verb|m2.input|, \verb|m2.output|), while navigating the
        internal structures of the DAG, i.e. the Node and Edge EDB, we provide a regexp style
        \emph{selector operator} on a model version to access individual \dnn\ nodes\eat{. For example,}, e.g. 
        \verb|m1["conv[1,3,5]"]| in Query~\ref{dql:select} filters the nodes in \verb|m1|. 
        Once the selector operator returns a set of nodes, \verb|prev| and \verb|next| attributes of
        the node allow 1-hop traversal in the DAG. Note that \verb|POOL("MAX")| is one of the
        standard built-in node templates for condition clauses. Using \emph{SPJ} operators with
        object type \emph{attribute access} and the \emph{selector operator}, we allow 
        relational queries to be mixed with graph traversal conditions.

To retrieve reusable components in a DAG, and mutate it to get new models, we provide \textbf{slice},
   \textbf{construct} and \textbf{mutate} operators. \textbf{Slice} originates in programming analysis
   research; given a start and an end node, it
    returns a subgraph including all paths from the start to the end and the connections which
    are needed to produce the output. \textbf{Construct} can be found in graph query
    languages such as SPARQL to create new graphs. \eat{In our context, the \dnn\ DAG only has nodes
    with multiple attributes, which simplifies the language. }We allow \textbf{construct} to
    derive new DAGs by using selected nodes to \emph{insert} nodes by splitting an outgoing edge or
    to \emph{delete} an outgoing edge connecting to another node. \textbf{Mutate} limits
    the places where \emph{insert} and \emph{delete} can occur. For example, Query~\ref{dql:slice}
    and~\ref{dql:construct}\eat{ show queries which work on the DAG structure and} generate reusable
    subgraphs and new graphs. Query~\ref{dql:slice} slices a sub-network from matching models
    between convolution layer `conv1' and full layer `fc7', while Query~\ref{dql:construct} derives
    new models by appending a ReLU layer after all convolution layers followed by an average pool. All queries can be nested.\eat{ in the \emph{from} clause. }


Finally, {\bf evaluate} can be used to try out new models, with potential for early out if
expectations are not reached. 
We separate the network enumeration component from the hyper-parameter turning component; 
while network enumeration can be \eat{done via
\emph{select} or \emph{construct} and }nested in the \emph{from} clause, we introduce a \emph{with operator}
to take an instance of a tuning config template, and a \emph{vary operator} to express the
combination of activated multi-dimensional hyper-parameters and search strategies. \emph{auto} is
keyword implemented using default search strategies (currently grid search). To stop early and let
the user control the stopping logic, we introduce a \emph{keep operator} to take a rule consisting of 
stopping condition templates, such as top-k of the evaluated models, or accuracy threshold.
Query~\ref{dql:eval} evaluates the models constructed and tries combinations of at least three
different hyper-parameters, and keeps the top 5 models w.r.t. the loss after 100 iterations.

\eat{
Besides the query facilities we have described so far, there are a collection of features we are lack of space to describe in details, such as extraction and interactions with model metadata, weight management, and provenance queries of the models.

\subsection{Model Publishing \& Sharing}
As the model repository is standalone, we host the repositories as a whole in a \modelhub\ service. The modeler can use the \dlvcmd{publish} to push the repository for archiving, collaborating or sharing, and use \dlvcmd{search} and \dlvcmd{pull} to discover and reuse remote models. We envision such a form of collaboration can facilitate a learning environment, as all versions in the lifecycle are accessible and understandable with ease.
}

\subsection{ModelHub Implementation} 
On the local side, the current implementation of \modelhub\ maintains the data model in multiple back-ends and utilizes \cmd{git} to manage the arbitrary file diffs. Various queries are decomposed and sent to different backends and chained accordingly. 
On the other hand, as the model repository is standalone, we host the repositories as a whole in a \modelhub\ service. The modeler can use the \dlvcmd{publish} to push the repository for archiving, collaborating or sharing, and use \dlvcmd{search} and \dlvcmd{pull} to discover and reuse remote models. We envision such a form of collaboration can facilitate a learning environment, as all versions in the lifecycle are accessible and understandable with ease.


\icdeadjust{
\section{Model Comparison}
\label{subsec:comparison}
The subtle difference of model versions are hard to grasp, comparing models are time consuming and requires heavy scripting under current training systems. As accuracy is the goal of the whole process, often modelers judge models by simple performance measures (e.g. loss, accuracy). Shared scripts such as plotting architectures, and optimization training logs can be found in user community of specific training systems.

\modelhub\ supports a set of comparison schemes for different artifacts shown in Table~\ref{tb:dlv_comparison}. One can tell the different flavor of \cmd{diff} operations from mixture of data types. Besides common practice, new proposed schemes are highlighted in italic font.

\begin{table}[!t]
\centering
\begin{tabular}{ccccc} 
\toprule
\textbf{Artifact} & \textbf{Type} & \textbf{Scheme} 
\\ 
\midrule
Network Architecture & Graph & Plot, \emph{Graph Edit Distance} \\
\midrule
Learned Parameters & Tensor & Plot, \emph{Matrix Reordering} \\
\midrule
Prediction Result & Relational & Set Diff \\
\midrule
Hyperparameters & Key-Value & Set Diff \\
\midrule
Optimization Routine & Time Series & Plot \\
\bottomrule
\end{tabular}
\caption{\dlvcmd{diff}\ \ Comparison Schemes}
\label{tb:dlv_comparison}
\end{table}

Plotting is an important tool to understand a network architecture, while comparing the difference between two networks architectures or two tensors are not easy. By viewing a network architecture as a DAG in \modelhub, the comparison can be done via a minimizing graph edit distance routine. For tensor comparison, we propose a new matrix reordering scheme. 

\subsection{Align Network Architectures}
\huicomment{}

\subsection{Align Learned Parameters}
For ease of illustration, we focus on 2D matrices, as tensors in \dnn\ can be lowered from high dimension to connections between input and output neurons.
Informally, the basic comparison idea is given two matrices $A$ and $B$, we permute $B$'s rows and columns accordingly in order to find the most similar $B'$ w.r.t. to a cost function, e.g. euclidean distance, best compression bits. As an example, the direct and best delta matrix is shown in Fig.~\ref{fig:alignment_examples}. In Fig.~\ref{fig:alignment_fix_init_example}, we show the weight matrices of LeNet conv1 layer trained with the same initialization but flipped images for the scenario of weight reusing. In Fig.~\ref{fig:alignment_random_init_example}, the LeNet is trained with random initialization with the same image orders. As we can see, alignment not only show the connections of two matrices, but also derive more zeros. As shown later in the evaluation, together with segmented float matrices, alignment operator contributes to the overall storage performance.

\begin{figure}[!t]
\subfigure[Fixed Initialization]{
  \includegraphics[width=1.0\linewidth]{fixed_init_alignment.pdf} 
  \label{fig:alignment_fix_init_example}
}
\subfigure[Random Initialization]{
  \includegraphics[width=1.0\linewidth]{random_init_alignment.pdf}
  \label{fig:alignment_random_init_example}
}
\caption{Example of Direct And Best Aligned Delta. LeNet conv1: {\small{(LeNet A, B, Direct Delta A-B, Best Delta A-B')}}\huicomment{todo}}
\label{fig:alignment_examples}
\end{figure}

Next, we present the matrix alignment problem formally. As the matrices to be aligned not necessary having the same dimensions, we first define a \emph{permutation matrix} with capability to adapt dimensions.

\begin{definition}[Permutation Matrix]
Let a permutation $\mathbb{\psi} =$\\ $ (\psi_1, \psi_2, \cdots, \psi_n) \in \Psi$, $\Psi$ is all possible $n!$ permutations. Given two positive integer $s,t\in \mathbb{N}_+$, a permutation matrix of $\psi$ is a matrix $\mathcal{P}(\psi,s,t) \in [0,1]^{s \times t}$, where
\begin{center}
$\mathcal{P}(\psi,s,t)_{i,j} = \begin{cases}
1 &\text{when\ } i \leq n, \psi_i = j\\
0 &\text{otherwise}
\end{cases}$
\end{center}
\end{definition}

Using the permutation matrix, a permutation can be used in matrix multiplication to reorder matrix by row or column accordingly.

\begin{example}
A row permutation $\psi = (2, 1)$ of $A \in \mathbb{R}^{2 \times 4}$ to $\mathbb{R}^{3\times 4}$:
\begin{center}
{\small{
$\mathcal{P}(\psi,3,2) \times A = 
\begin{bmatrix} 
0 & 1 \\
1 & 0 \\
0 & 0 
\end{bmatrix} \times
\begin{bmatrix} 
8 & 2 & 5 & 3\\
4 & 1 & 2 & 2
\end{bmatrix}
=
\begin{bmatrix}  
4 & 1 & 2 & 2\\
8 & 2 & 5 & 3\\
0 & 0 & 0 & 0
\end{bmatrix}
$
}}
\end{center}

A column permutation $\phi = (1, 3, 2, 4)$ of $A \in \mathbb{R}^{2 \times 4}$ to $\mathbb{R}^{2\times 3}$:
\begin{center}
$A \times \mathcal{P}(\phi,4,3) = 
\begin{bmatrix} 
8 & 2 & 5 & 3\\
4 & 1 & 2 & 2
\end{bmatrix} \times
\begin{bmatrix} 
1 & 0 & 0  \\
0 & 0 & 1  \\
0 & 1 & 0  \\
0 & 0 & 0  
\end{bmatrix} =
\begin{bmatrix}  
8 & 5 & 2 \\
4 & 2 & 1 
\end{bmatrix}
$
\end{center}
\end{example}

With the permutation matrix, given two matrices and a cost function, we formulate the matrix alignment problem as follows:

\begin{problem}[Matrix Alignment]
Given two real matrices, $A \in \mathbb{R}^{m\times n}$, $B \in \mathbb{R}^{s \times t}$, we want to find two permutations to reorder $B$ to $A$, $\psi = (\psi_1, \psi_2, \cdots, \psi_s) \in \Psi_{row}$ and $\phi = (\phi_1, 
\phi_2, \cdots, \phi_t) \in \Phi_{col}$, such that:
\begin{displaymath}
\psi^*, \phi^* = \argmin_{\psi \in \Psi_{row}, \phi \in \Phi_{col}} \mathcal{C}(A - \mathcal{P}(\psi,m,s)\times B \times \mathcal{P}(\phi,t,n))
\end{displaymath}
where $\mathcal{C}: \mathbb{R}^{m \times n} \mapsto \mathbb{R}$ is a cost function. We denote the best delta matrix as $\Delta^*_{A,B} = A - \mathcal{P}(\psi^*,m,s)\times B \times \mathcal{P}(\phi^*,t,n)$.
\end{problem}

\begin{example}
Let $\mathcal{C}$ be $\left\|.\right\|_2$ norm, given two following matrices:
\begin{center}
{\small{
$A = \begin{bmatrix} 
8 & 2 & 5 & 3\\
4 & 1 & 2 & 2
\end{bmatrix}, B = \begin{bmatrix}
1 & 1 & 1\\
8 & 5 & 2\\
4 & 2 & 1\\
2 & 2 & 2 
\end{bmatrix}$
}}, \\ 
\end{center}
the permutations
$\psi^* = (2,3,1,4), \phi^*=(1,3,2)$ minimize the cost function:
$\left\|A - \mathcal{P}(\psi^*,2,4)\times B \times\mathcal{P}(\phi^*,3,4))\right\|_2$=\\
{\scriptsize{
$
\left\|\begin{bmatrix} 
8 & 2 & 5 & 3 \\
4 & 1 & 2 & 2 
\end{bmatrix} - 
\begin{bmatrix} 
0 & 1 & 0 & 0 \\
0 & 0 & 1 & 0 
\end{bmatrix}
\begin{bmatrix} 
1 & 1 & 1 \\
8 & 5 & 2 \\
4 & 2 & 1 \\
2 & 2 & 2 
\end{bmatrix}
\begin{bmatrix} 
1 & 0 & 0 & 0 \\
0 & 0 & 1 & 0 \\
0 & 1 & 0 & 0 \\
\end{bmatrix}\right\|_2
=
\left\|\begin{bmatrix} 
0 & 0 & 0 & 3 \\
0 & 0 & 0 & 2 \\
\end{bmatrix}\right\|_2$
}}
\label{exp:alignment_problem}
\end{example}


\topic{Complexity Analysis}:
To the best of our knowledge, the matrix alignment problem is not studied in the literature. We show its NP-completeness by using the graph edit distance problem (\ged).

\begin{theorem}
Matrix Alignment Problem is NP-Complete when $\mathcal{C}$ is additive.  
\end{theorem}

\eat{
\begin{figure}[!t]
\centering
\includegraphics[width=1.0\linewidth]{snapshots.pdf}
\caption{Weight Artifacts in a DLV Repository}
\label{fig:snapshots}
\end{figure}

\begin{figure}[!t]
\centering
\includegraphics[width=1.0\linewidth]{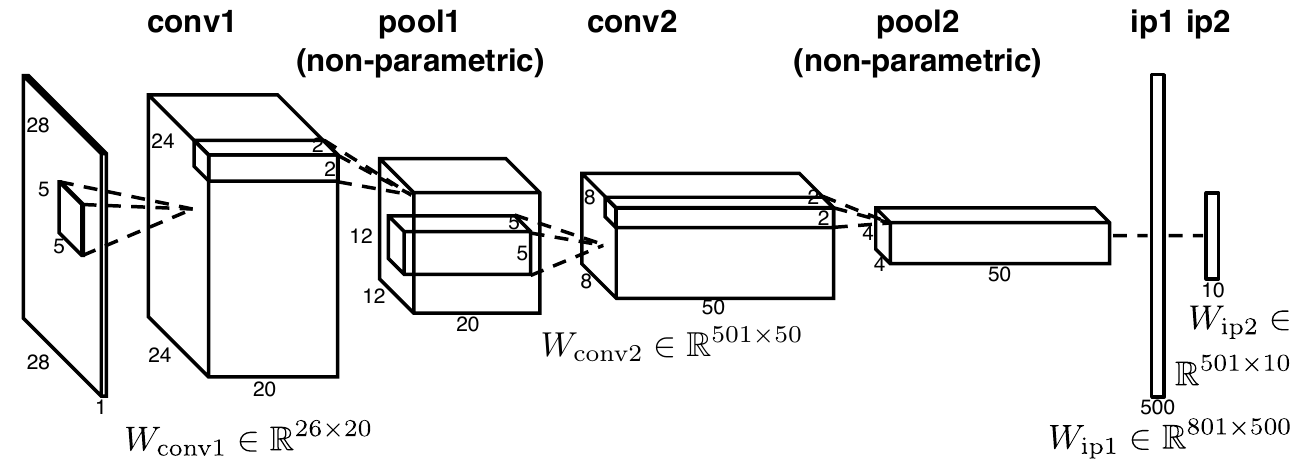}
\caption{One Snapshot Example  (LeNet)}
\label{fig:onesnapshot}
\end{figure}
}

\subsubsection{Greedy Hill Climbing Algorithm}

\begin{algorithm}[b!] 
\caption{Randomized Hill Climbing Algorithm}
\label{alg:greedy_alignment}
\begin{algorithmic}[1]
\REQUIRE Matrix $A, B$, Random Initial Points $r$, Max Iterations $k$
\ENSURE Local optimum ($\psi^*, \phi^*$) of the Matrix Alignment
\STATE $i = 0$; $\psi=(1,2,\cdots,m)$; $\phi=(1,2,\cdots,n)$
\STATE result = ($\psi, \phi$); cost = $\left\|A-\mathcal{P}(\psi,m,s)\times B \times \mathcal{P}(\phi,t,n)\right\|_2^2$
\WHILE{$ i < r$}
  \STATE random permutation $\psi, \phi$; $c = \infty$
  \FOR{$j=0; j < k; j$++}
    \STATE $\mathcal{B}_r =$ Construct\_RAB(A, B, $\psi,\phi$)
    \STATE $E_M^*, c =$ Maximum\_Weighted\_Matching($\mathcal{B}_r$)
    \STATE $\psi =$ Derived\_Row\_Permutation($E_M^*$)
    \STATE $\mathcal{B}_c =$ Construct\_CAB(A, B, $\psi,\phi$)
    \STATE $E_M^*, c =$ Maximum\_Weighted\_Matching($\mathcal{B}_c$)
    \STATE $\phi =$ Derived\_Column\_Permutation($E_M^*$)
  \ENDFOR
  \IF{$c <$ cost}
  \STATE result = ($\psi, \phi$); cost = $c$
  \ENDIF
  \STATE $i$++
\ENDWHILE
\RETURN result
\end{algorithmic}
\end{algorithm}



We propose a randomized hill climbing approach to address the matrix alignment problem. As an overview, by noticing that if we fix one permutation at a time, minimizing the cost by varying the other permutation is the same as a graph matching problem in a biclique, we can iteratively solve a series of maximum weighted bipartite matching problems and find a local optimum. With randomized initial permutation, we run the algorithm multiple times, and choose the best solution among the found local optima.

Now we illustrate the algorithm in detail. First, fixing one permutation at a time, e.g. the column permutation $\phi = \bar \phi$, the matrix alignment problem is finding the best row permutation such that: 
\begin{displaymath}
\psi^* = \argmin_{\psi \in \Psi_{row}} \mathcal{C}(A - \mathcal{P}(\psi,m,s)\times B \times \mathcal{P}(\bar \phi,t,n))
\end{displaymath}

Next, we define the row alignment biclique (RAB), and show the connections between its maximum weighted bipartite matching solution and the best row alignment $\psi^*$ in detail.

\begin{definition}[Row Alignment Biclique] 
Given two real matrices, $A \in \mathbb{R}^{m\times n}$, $B \in \mathbb{R}^{s \times t}$, a row alignment biclique is a bipartite graph $\mathcal{B}_r(V_1, V_2, E, w_e)$, where $V_1$ is a set of vertices representing all row vectors of $A$, $|V_1| = m$, and $v_{1,i} \in V_1$ is the i-{th} row $r_{A,i}$ of $A$; while $V_2$ represents the row vectors of the matrix:
\begin{displaymath}
B' = \mathcal{P}((1,2,\cdots,s), \max\{m,s\}, s)\times B\times \mathcal{P}((1,2,\cdots,t), t,n),
\end{displaymath}
$|V_2| = \max\{m,s\}$ and $v_{2,j} \in V_2$ is the j-{th} row $r_{B',j}$ of $B'$. Each $e_{i,j} = (v_i, v_j) \in E$ connects a row pair between $V_1$ and $V_2$, and $w_e: e \mapsto - \mathcal{C}(r_{A,i}, r_{B',j})$ is the weight, i.e. $- \left\|r_{A,i} - r_{B',j}\right\|_2^2$. 
\end{definition}

\begin{figure}[h!]
\centering
\includegraphics[width=0.7\linewidth]{row_permutation_exp.pdf}
\caption{Example of Row Alignment Biclique}
\label{fig:row_permutation_exp}
\end{figure}

Given $A \in \mathbb{R}^{m\times n}$ and $B \in \mathbb{R}^{s\times t}$, a \emph{maximal bipartite matching} $E_M$ of the RAB $\mathcal{B}_r(V_1, V_2, E, w_e)$ is a set of edges sharing no vertices and $|E_M| = m$. Given an $E_M$, we can get a \emph{derived row permutation} $\psi = (\psi_1, \psi_2, \cdots, \psi_s) \in \Psi_{row}$ as follows:
\begin{displaymath}
\psi_i = \begin{cases}
j &\text{when\ } (v_{1,i}, v_{2,j}) \in E_M\\
0 &\text{otherwise} 
\end{cases}
\end{displaymath}

\emph{The maximum weighted bipartite matching} $E^*_M$ is a maximal bipartite matching with the biggest $\sum_{e_i\in E_M} w_e(e_i)$ among all maximal bipartite matchings $\{E_M\}$.

\begin{example} In Fig.~\ref{fig:row_permutation_exp}, we use the two matrices $A$, $B$ from Exp.~\ref{exp:alignment_problem} to show its row alignment biclique. $B'$ is the resized matrix of $B$ by applying two permutation matrices $\mathcal{P}((1,2,3,4), 4, 4)$ and $\mathcal{P}((1,2,3),\\ 3, 4)$. The weight of each edge is shown in the figure. For instance, $w_e(v_{1,1}, v_{2,2}) = - ((8-8)^2 + (2-5)^2 + (5-2)^2 + (3-0)^2) = -27$.

The maximal matching $E_M^* = \{(v_{1,1}, v_{2,2}), (v_{1,2}, v_{2,3})\}$ with the sum of weight $-33$ is the maximum weighted bipartite matching. The derived permutation is $(2,3,0,0)$.
\end{example}


\begin{lemma} The derived row permutation $\psi^*$ of the maximum weighted bipartite matching $E^*_M$ is also the solution of the matrix alignment problem with fixed column permutation $\bar \phi$.
\end{lemma}
\begin{proof}
As $E^*_M$ has the largest $- \sum_{e} \left\|r_{A,i} - r_{B',j}\right\|_2^2$, the derived permutation $\psi^*$ also has the minimum $\left\|A - \mathcal{P}(\psi^*,m,s)\times B \times \mathcal{P}(\bar \phi,t,n)\right\|_2^2$.
\end{proof}

Due to the symmetric definition of the matrix alignment problem, we can define column alignment biclique similarly, as well as its maximum weighted bipartite matching. The derived permutation $\phi^*$ is also the solution for the matrix alignment problem with fixed row permutation $\psi^*$.

\begin{lemma} Iterative maximum weighted bipartite matching in Alg.~\ref{alg:greedy_alignment} converges to a local optimum solution.  
\end{lemma}

The full algorithm is described in Alg.~\ref{alg:greedy_alignment}, $k$ is the max iteration, $r$ is the random initial points. By applying the bipartite matching iteratively, i.e. fixing a row or column permutation with previous iteration result, the cost is monotonically decreasing and we can find a local optima of the matrix alignment problem. In other words, it is a hill climbing algorithm to find a local optimum. With random initial permutations, we can compute multiple local optima and choose the best one. 
}{no space for this}

\section{Parameter archival storage (\weightstore)}
\label{sec:binary_storage}


\begin{figure}[!t]
\centering
\includegraphics[width=0.9\linewidth]{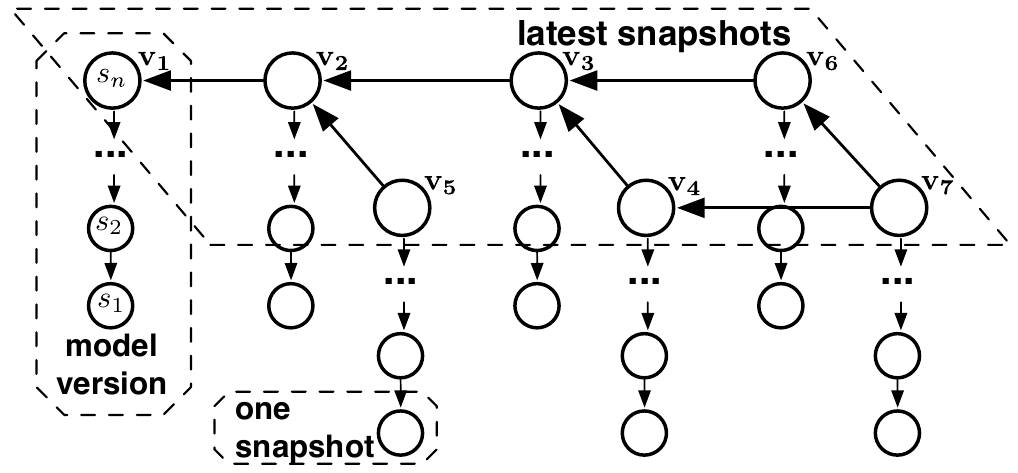}
\caption{Relationships of Model Versions and Snapshots}
\label{fig:snapshots}
\end{figure}


Modeling lifecycle for DNNs, and machine learning models in general, is centered around the learned parameters, whose storage footprint can be very large. 
The goal of \weightstore\ is to maintain a large number of learned models as compactly as possible, without compromising on the query performance.
Before introducing our design, we first discuss the queries of interest, and some key properties of the model artifacts. 
We then describe different options to store a single {\em float matrix}, and to construct {\em deltas (differences)} between two matrices.
We then formulate the optimal version graph storage problem, discuss how it differs from the prior work, and present algorithms for solving it.
Finally, we develop a novel approximate model evaluation technique, suitable for the segmented storage technique that \weightstore\ uses.
%

\subsection{Weight Parameters \& Query Type of Interests}
\label{subsec:pas_query_interests}
We illustrate the key weight parameter artifacts and the relationships among them in Fig.~\ref{fig:snapshots}, and also explain some of the notations
used in this section.
At a high level, the predecessor-successor relationships between all the developed models is captured as a {\bf version graph}. These relationships are user-specified and 
conceptual in nature, and the interpretation is left to the user (i.e., an edge $v_i \rightarrow v_j$ indicates that $v_j$ was an updated version of the model that the user
checked in after $v_i$, but the nature of this update is irrelevant for storage purposes). 
A model version $v_i$ itself consists of a series of snapshots, $s_1, ..., s_n$, which represent checkpoints during the training process (most systems will take such snapshots
due to the long running times of the iterations).
We refer the last or the best checkpointed snapshot $s_n$ as the \textbf{latest snapshot} of $v_i$, and denote it by $s_{v_i}$. 

One snapshot, in turn, consists of intermediate data $X$ and trained parameters $W$ (e.g., in Fig.~\ref{fig:onesnapshot}, the model has $431080$ parameters for $W$, and $19694\cdot b$ dimensions for $X$, where $b$ is the minibatch size). Since $X$ is useful only if training needs to be resumed, only $W$ is stored in \weightstore. 
Outside of a few rare exceptions, $W$ can always be viewed as a collection of float matrices, $\mathbb{R}^{m\times n}, m \geq 1, n \geq 1$, which encode the weights on the edges from outputs of the neurons in one layer to the inputs of the neurons in the next layer. Thus, we treat a {\em float matrix} as a first class data type in \weightstore\footnote{We do not make a distinction about the bias weight; the typical linear transformation $W'x+b$ is treated as $W\cdot(x,1) = (W',b)^T\cdot(x,1)$.}.

The retrieval queries of interest are dictated by the operations that are done on these stored models, which include: 
(a) testing a model, (b) reusing weights to fine-tune other models, (c) comparing parameters of different models, (d) comparing the results of different models on a dataset, and (e) model exploration queries (Sec.~\ref{subsec:query}). Most of these operations require execution of {\bf group retrieval} queries, where all the weight matrices in a specific snapshot need to be retrieved. 
This is different from range queries seen in array databases (e.g., SciDB), and also have unique characteristics that 
influence the storage and retrieval algorithms. 

\begin{list}{$\bullet$}{\leftmargin 0.10in \topsep -2pt} 
\item \emph{Similarity among Fine-tuned Models}: Although non-convexity of the training algorithm and differences in network architectures across models lead to non-correlated parameters, the widely-used fine-tuning practices (Sec.~\ref{sec:preliminary}) generate model versions with similar parameters, resulting in efficient delta encoding schemes.
\item \emph{Co-usage constraints:} Prior work on versioning and retrieval~\cite{vldb15versioning} has focused on retrieving a single artifact stored in its entirety. However, we would like to store the different matrices in a snapshot independently of each other, but we must retrieve them together. These co-usage constraints make the prior algorithms inapplicable as we discuss later.
\item \emph{Low Precision Tolerance}: DNNs are well-known for their tolerance to using low-precision floating point numbers (Sec.~\ref{sec:related_work}), both during training and evaluation. Further, many types of queries (e.g., visualization and comparisons) do not require retrieving the full-precision weights.
\item \emph{Unbalanced Access Frequencies}: Not all snapshots are used frequently. The latest snapshots with the best testing accuracy are used in most of the cases. The checkpointed snapshots have limited usages, including debugging and comparisons. 
\eat{\weightstore\ provides a set of storage schemes to let the user trade-off between storage and lossyness.}
\end{list}

\subsection{Parameters As Segmented Float Matrices}
\label{subsec:pas_float}

\topic{Float Data Type Schemes}: 
Although binary (1/-1) or ternary (1/0/-1) matrices are sometimes used in DNNs, in general \weightstore\ handles real number weights. Due to different usages of snapshots, \weightstore\ offers a handful of float representations to let the user trade-off storage efficiency with lossyness using \dlv. 
\begin{list}{$\bullet$}{\leftmargin 0.10in \topsep -2pt} 
\item \emph{Float Point}: DNNs are typically trained with single precision (32 bit)\eat{or less likely double precision (64 bit)} floats. This scheme uses the standard IEEE 754 floating point encoding to store the weights with sign, exponent, and mantissa bits. IEEE half-precision proposal (16 bits) and tensorflow truncated 16bits~\cite{tensorflow} are supported as well and can be used if desired.
\item \emph{Fixed Point}: \eat{Comparing with float point encoding where each float has exponent bits, f}Fixed point encoding has a global exponent per matrix, and each float number only has sign and mantissa using all $k$ bits. This scheme is a lossy scheme as tail positions are dropped, and a maximum of $2^k$ different values can be expressed. The entropy of the matrix also drops considerably, aiding in compression.
\item \emph{Quantization}: Similarly, \weightstore\ supports quantization using $k$ bits, $k\leq8$, where $2^k$ possible values are allowed. The quantization can be done in random manner or uniform manner by analyzing the distribution, and a coding table is used to maintain the \eat{quantization information (with only the} integer codes stored in the matrices in \weightstore\eat{)}.
This is most useful for snapshots whose weights are primarily used for fine-tuning or initialization.
\end{list}

The float point schemes present here are not new, and are used in DNN systems in practice~\cite{vanhoucke2011improving,han2015deep,courbariaux2014training}.\eat{, focusing on their implications of training/testing phases.} As a lifecycle management tool, \weightstore\ lets experienced users select schemes rather than deleting snapshots due to resource constraints. Our evaluation shows storage/accuracy tradeoffs of these schemes.

\eat{
\begin{table}[!t]
{\small{
\centering
\begin{tabular}{ccccc} 
\toprule
\textbf{Scheme} & \textbf{Param. Bits} & \textbf{Compress} & \textbf{Lossyness} & \textbf{Usage}\\ 
\midrule
Float Point  & 64/32/16 & Fair & Lossless & latest \\ 
\midrule
Fixed Point  & 32/16/8 & Good & Good  & latest \\ 
\midrule
Quantization & 8/k & Excellent & Poor & other\\ 
\bottomrule
\end{tabular}
}}
\caption{Float Representation Scheme Trade-offs}
\label{tb:lowprecision_offering}
\end{table}
}

\topic{Bytewise Segmentation for Float Matrices}: 
One challenge for \weightstore\ is the high entropy of float numbers in the float arithmetic representations, which leads to them being very hard to  compress. 
Compression ratio shown in related work for scientific float point datasets, e.g., simulations, is very low. 
The state of art compression schemes do not work well for \dnn\ parameters either (Sec.~\ref{sec:related_work}).
\eat{
\begin{table}[!h]
\centering
{\small{
\begin{tabular}{cccc} 
\toprule
\textbf{Method} & \textbf{LeNet} & \textbf{AlexNet} & \textbf{VGG16} \\ 
\midrule
fpzip  &  &  & \huicomment{to add}  \\ 
\midrule
S-Wavlet  &  &  &    \\ 
\midrule
ISOBAR &  &  &  \\ 
\bottomrule
\end{tabular}
}}
\caption{Compression Ratio for \dnn\ Parameters}
\label{tb:compression_tech}
\end{table}
}
By exploiting \dnn\ low-precision tolerance, we adopt bytewise decomposition from prior work~\cite{schendel2012isobar,bhattacherjee2014pstore} and extend it to our context to store the float matrices. 
The basic idea is to separate the high-order and low-order mantissa bits, and so a float matrix is stored in multiple chunks; the first chunk consists of 8 high-order bits, and the rest are segmented one byte per chunk. 
One major advantage is the high-order bits have low entropy, and standard compression schemes (e.g., \emph{zlib}) are effective for them.

Apart from the simplicity of the approach, the key benefits of segmented approach are two-fold: (a) it allows offloading low-order bytes to remote storage, (b) \weightstore\ queries can read high-order bytes only, in exchange for tolerating small errors.
Comparison and exploration queries (\dlvcmd{desc}, \dlvcmd{diff}) can easily tolerate such errors and, as we show in this paper, \dlvcmd{eval} queries can also be made tolerant to these errors.

\eat{
\begin{table}[!h]
{\small{
\centering
\begin{tabular}{ccccc} 
\toprule
\textbf{Read (bits)} & Float (64) & Float (32) & Fix (k/n) & Fix (32/n)
\\ 
\midrule
16 & & & & \huicomment{todo}\\ 
\midrule
24 & & & & \\ 
\midrule
32 & & & & \\ 
\bottomrule
\end{tabular}
}}
\caption{Maximum Absolute Error With High-Order Bytes}
\label{tb:bytewise_error}
\end{table}
}

\topic{Delta Encoding Across Snapshots}: 
We observed that, due to the non-convexity in training, even re-training the
same model with slightly different initializations results in very different
parameters. 
However, the parameters from checkpoint snapshots for the same
or similar models tend to be close to each other. 
%
Furthermore, across model versions, fine-tuned models generated using fixed
initializations from another model often have similar parameters. The
observations naturally suggest use of {\em delta encoding} between checkpointed snapshots
in one model version and latest snapshots across multiple model versions; i.e., instead of storing
all matrices in entirety, we can store some in their entirety and others as differences from those.
Two possible delta functions (denoted $\cmstdeltaop$) are {\em arithmetic subtraction} and {\em bitwise XOR}\footnote{
Delta functions for matrices with different dimensions are discussed in the long version of the paper; 
techniques in
Sec~\ref{sec:binary_storage} work with minor modification.}. 
We find the compression footprints when applying
the diff $\cmstdeltaop$ in different directions are similar.
We study the delta operators on real models in Sec.~\ref{sec:experiments}. 


\begin{figure}[!t]
\subfigure[Matrix Storage Graph {\newline}]{
\includegraphics[width=0.325\linewidth]{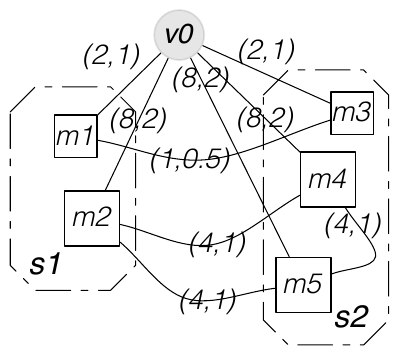}
\label{fig:storage_graph}
}\subfigure[Optimal Plan {\newline} without Constraints]{
\includegraphics[width=0.295\linewidth]{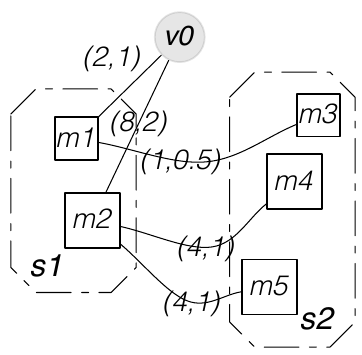}
\label{fig:storage_plan_mst}
}\subfigure[Optimal Plan with {\newline} {\scriptsize{$\cmstvcr^{\cmstscheme{i}}(s_1)\leq3 \wedge \cmstvcr^{\cmstscheme{i}}(s_2)\leq6$}}]{
\includegraphics[width=0.295\linewidth]{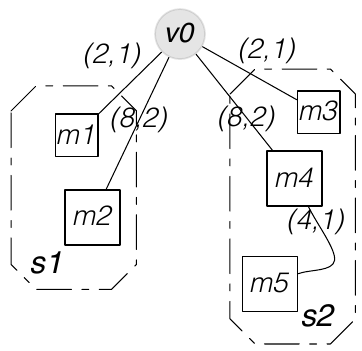}
\label{fig:storage_plan_constrained}
}
\caption{Illustrating Matrix Storage Graph \& Plan using a Toy Example}
\label{fig:optimal_archival_storage_exp}
\end{figure}

\subsection{Optimal Parameter Archival Storage}
\label{subsec:pas_version}

Given the above background, we next address the question of how to best store a collection of model versions, so that the total storage footprint occupied by the large segmented float matrices is minimized while the retrieval performance is not compromised. This recreation/storage tradeoff sits at the core of any version control system. In recent work~\cite{vldb15versioning}, the authors study six variants of this problem, and show the NP-hardness of most of those variations. 
However, their techniques cannot be directly applied in \weightstore, primarily because their approach is not able to handle the {\em group retrieval (co-usage)} constraints.



We first introduce the necessary notation, discuss the differences from prior work, and present the new techniques we developed for \weightstore.
In Fig.~\ref{fig:snapshots}, a \emph{model version} $v \in V$ consists of time-ordered checkpointed \emph{snapshots}, $S_v = {s_1, ..., s_n}$. Each \emph{snapshot}, $s_i$ consists of a named list of float matrices $M_{v,i} = \{m_k\}$ representing the learned parameters. All matrices in a repository, $\cmstallmatrix = \bigcup_{v \in V} \bigcup_{s_i \in S_v} { M_{v,i} } $, are the parameter artifacts to archive. Each matrix $m \in \cmstallmatrix$ is either stored directly, or is recovered through another matrix $m' \in \cmstallmatrix$ via a delta operator $\cmstdeltaop$, i.e. $m = m'\ \cmstdeltaop\ d$, where $d$ is the delta computed using one of the techniques discussed above. In the latter case, the matrix $d$ is stored instead of $m$. To unify the two cases, we introduce a empty matrix $\cmstroot$, and define $\forall \cmstdeltaop\forall m \in \cmstallmatrix, m\ \cmstdeltaop\ \cmstroot = m$.

\begin{definition}[Matrix Storage Graph] Given a repository of model versions $V$, let $\cmstroot$ be an empty matrix, and $\cmstvgvset = \cmstallmatrix \cup \{\cmstroot\}$ be the set of all parameter matrices. We denote by $\cmstvgeset = \{m_i\ \cmstdeltaop\ m_j \} \cup \{ m_i\ \cmstdeltaop\ \cmstroot\}$ the available deltas between all pairs of matrices. Abusing notation somewhat, we also treat $\cmstvgeset$ as the set of all {\em edges} in a graph where $\cmstvgvset$ are the vertices. Finally, let $\cmstvg$ denote the {\em matrix storage graph} of $V$, 
where edge weights $\cmstecs, \cmstecr: \cmstvgeset \mapsto \mathbb{R^+}$ are storage cost and recreation cost of an edge respectively.
\end{definition}

\begin{definition}[Matrix Storage Plan] Any connected subgraph of $\cmstvgsimp$ is called a matrix storage plan for $V$, and denoted by $\cmstplan$, where $\cmstvgvset_P = \cmstvgvset$ and $\cmstvgeset_P \subseteq \cmstvgeset$.
\end{definition}

\begin{example}
In Fig.~\ref{fig:storage_graph}, we show a matrix storage graph for a repository with two snapshots, $s_1 =\{m_1,m_2\}$ and $s_2 =\{m_3,m_4,m_5\}$. The weights associated with an edge $e = (\cmstroot, m_i)$ reflect the cost of materializing the matrix $m_i$ and retrieving it directly. On the other hand, for an edge between two matrices, e.g., $e = (m_2, m_5)$, the weights denote the storage cost of the corresponding delta and the recreation cost of applying that delta. In Fig.~\ref{fig:storage_plan_mst} and \ref{fig:storage_plan_constrained}, two matrix storage plans are shown. 
\end{example}

For a matrix storage plan $\cmstplan$, \weightstore\ stores all its edges and is able to recreate any matrix $m_i$ following a path starting from $\cmstroot$. The \underline{\emph{total storage cost}} of $\cmstplansimp$, denoted as $\cmstvcs(\cmstplansimp)$, is simply the sum of edge storage costs, i.e.
$\cmstvcs(\cmstplansimp) = \sum_{e\in\cmstvgeset_P} \cmstecs(e)$. 
Computation of the  \underline{\emph{average snapshot recreation cost}} is more involved and depends on the retreival scheme used:
\begin{list}{$\bullet$}{\leftmargin 0.10in \topsep -2pt} 
\item \emph{Independent} scheme recreates each matrix $m_i$ one by one by following the shortest path ($\cmstspath{m_i}$) to $m_i$ from $\cmstroot$. In that case, the recreation cost is simply computed by summing the recreation costs for all the edges along the shortest path.

\item \emph{Parallel} scheme accesses all matrices of a snapshot in parallel (using multiple threads); the longest shortest path from $\cmstroot$ defines the recreation cost for the snapshot. 
\item \emph{Reusable} scheme considers caching deltas on the way, i.e., if paths from $\cmstroot$ to two different matrices overlap, then the shared computation is only done once. In that case, we need to construct the lowest-cost {\em Steiner tree} ($\mathcal{T}_{\cmstplansimp, s_i}$) involving $\cmstroot$ and the matrices in the snapshot. However, because multiple large matrices need to be kept in memory simultaneously, the memory consumption of this scheme can be large.
\end{list}
\begin{table}[!h]
\centering
\begin{tabular}{cccc} 
\toprule
\emph{Retrieval Scheme} & \emph{Recreation} $\cmstvcr^{\cmstscheme{}}(\mathcal{P}_V, s_i)$ & \emph{Solution of Prob.1}
\\ 
\midrule
Independent ($\cmstscheme{i}$) & $\sum_{m_j \in s_i} \sum_{e_k \in \cmstspath{m_j}} \cmstecr(e_k) $ & Spanning tree \\ 
Parallel ($\cmstscheme{p}$) & $\max_{m_j \in s_i}\{\sum_{e_k \in \cmstspath{m_j}} \cmstecr(e_k)\}$ & Spanning tree  \\ 
Reusable ($\cmstscheme{r}$) & $\sum_{e_k \in \mathcal{T}_{\cmstplansimp, s_i}} \cmstecr(e_k)$ & Subgraph \\ 
\bottomrule
\vspace{0.1mm}
\end{tabular}
\caption{Recreation Cost of a Snapshot $s_i$ $\cmstvcr(\cmstplansimp,s_i)$ in a plan $\cmstplansimp$}
\label{tb:steiner_tree}
\end{table}

\weightstore\ can be configured to use any of these options during the actual query execution. 
However, solving the storage optimization problem with {\em Reusable} scheme is nearly impossible; since the Steiner tree problem is NP-Hard, just computing the cost of a solution becomes intractable making it hard to even compare two different storage solutions. Hence, during the storage optimization process, \weightstore\ can only support {\em Independent} or {\em Parallel} schemes. 


In the example above, the edges are shown as being undirected indicating that the deltas are symmetric. In general, we allow for directed deltas to handle asymmetric delta functions, and also for multiple directed edges between the same two matrices. The latter can be used to capture different options for storing the delta; e.g., we may have one edge corresponding to a remote storage option, where the storage cost is lower and the recreation cost is higher; whereas another edge (between the same two matrices) may correspond to a local SSD storage option, where the storage cost is the highest and the recreation cost is the lowest. Our algorithms can thus automatically choose the appropriate storage option for different deltas.


Similarly, \weightstore\ is able to make decisions at the level of byte segments of float matrices, by treating them as separate matrices that need to be retrieved together in some cases, and not in other cases. This, combined with the ability to incorporate different storage options, is a powerful generalization that allows \weightstore\ to make decisions at a very fine granularity. 


Given this notation, we can now state the problem formally. Since there are multiple optimization metrics, we assume that constraints on the retrieval costs are provided and ask to minimize the storage.


\begin{problem}[Optimal Parameter Archival Storage\eat{ Problem}]
\label{prob:pas_storage}
Given 
a matrix storage graph $\cmstvg$, let $\theta_i$ be the {\em snapshot recreation cost budget} for each $s_i \in S$. Under a retrieval scheme $\cmstscheme{}$, find a matrix storage plan $\cmstplansimp^*$ that minimizes the \emph{total storage cost}, while satisfying recreation constraints, i.e.:\\
{{\mbox{\ } \ $\minimize_{\cmstplansimp} \cmstvcs(\cmstplansimp); \quad \text{s.t.}\quad \forall s_i \in S, \cmstvcr^{\cmstscheme{}}(\cmstplansimp, s_i) \leq \theta_i$}}
\eat{
{\small{
\begin{equation*}
\minimize_{\cmstplansimp} \quad \cmstvcs(\cmstplansimp); \quad \text{s.t.}\quad \forall s_i \in S, \cmstvcr^{\cmstscheme{}}(\cmstplansimp, s_i) \leq \theta_i 
\end{equation*}
}}
}
\end{problem}

\begin{example}
In Fig.~\ref{fig:storage_plan_mst}, without any recreation constraints, we show the best storage plan, which is the minimum spanning tree based on $\cmstecs$ of the matrix storage graph, $\cmstvcs(\cmstplansimp) = 19$. Under independent scheme $\cmstscheme{i}$, $\cmstvcr^{\cmstscheme{i}}(\cmstplansimp, s_1) = 3$ and $\cmstvcr^{\cmstscheme{i}}(\cmstplansimp, s_2) = 7.5$. In Fig.~\ref{fig:storage_plan_constrained}, after adding two constraints $\theta_1 = 3$ and $\theta_2 = 6$, we shows an optimal storage plan $\cmstplansimp^*$ satisfying all constraints. The storage cost increases, $\cmstvcs(\cmstplansimp^*) = 24$, while $\cmstvcr^{\cmstscheme{i}}(\cmstplansimp^*, s_1) = 3$ and $\cmstvcr^{\cmstscheme{i}}(\cmstplansimp^*, s_2) = 6$. 
\end{example}

Although this problem variation might look similar to the ones considered in recent work~\cite{vldb15versioning}, none of the variations 
studied there can handle the co-usage constraints (i.e., the constraints on simultaneously retrieving a group of versioned data artifacts). One 
way to enforce such constraints is to treat the entire snapshot as a single data artifact that is stored together; however, that may force
us to use an overall suboptimal solution because we would not be able to choose the most appropriate delta at the level of individual matrices. 
Another option would be to sub-divide the retrieval budget for a snapshot into constraints on individual matrices in the snapshot. As our
experiments show, that can lead to significantly higher storage utilization. 
Thus the formulation above is a strict generalization of the formulations considered in that prior work. 

%

\begin{theorem}
Optimal Parameter Archival Storage Problem is NP-hard for all retrieval schemes in Table~\ref{tb:steiner_tree}.
\end{theorem}
\begin{proof}
We reduce Prob.5 in \cite{vldb15versioning} to the independent scheme $\cmstscheme{i}$, and Prob.6 in \cite{vldb15versioning} to the parallel scheme $\cmstscheme{p}$, by mapping each datasets as vertices in storage graph, and introducing a snapshot holding all matrices with recreation bound $\Theta_g$. For reuse scheme $\cmstscheme{r}$, it is at least as hard as weighted set cover problem if reducing a set to an edge $e$ with storage cost $\cmstecs(e)$ as weight, an item to an vertex in $\cmstvgsimp$, and set recreation budget $\Theta_g = \infty$.
\end{proof}

\begin{lemma}\label{lm:spanning_tree_solution}
The optimal solution for Problem~\ref{prob:pas_storage} is a spanning tree when retrieval scheme $\cmstscheme{}$ is \emph{independent} or \emph{parallel}.
\end{lemma}
\begin{proof}
Suppose we have a non-tree solution $\cmstplansimp$ satisfying the constraints, and also minimize the objective. Note that parallel and independent schemes are based on shortest path $\cmstspath{m}$ in $\cmstplansimp$ from $\cmstroot$ to each matrix $m$, so the union of each shortest path forms a shortest path tree. If we remove edges which are not in the shortest path tree from the plan to $\cmstplansimp'$, it results in a lower objective $\cmstvcs(\cmstplansimp')$, but still satisfying all recreation constraints, which leads to a contradiction.
\end{proof}
Lemma~\ref{lm:spanning_tree_solution} shows $\cmstplansimp^*$ is a spanning tree and connects our problem to a class of constrained minimum spanning tree problems. 
The above lemma is not true for the \emph{reusable} scheme ($\cmstscheme{r})$; snapshot Steiner trees satisfying different recreation constraints may share intermediate nodes resulting in a subgraph solution.


\topic{Constrained Spanning Tree Problem}:
In Problem~\ref{prob:pas_storage}, storage cost minimization while ignoring the recreation constraints 
leads to a minimum spanning tree (MST) of the storage matrix; whereas the snapshot recreation
constraints are best satisfied by using a shortest path tree (SPT). 
These problems are often referred to as
constrained spanning tree problems~\cite{deo1997computation} or shallow-light
tree constructions~\cite{vlsi1994optimal}, which have been studied in areas
other than dataset versioning, such as VLSI designs. 
Khuller et al.~\cite{khuller1995balancing} propose an algorithm called LAST to construct 
such a ``balanced'' spanning tree in an undirected graph $G$. LAST starts with a minimum spanning tree of the provided graph, 
traverses it in a DFS manner, and adjusts the tree by changing parents to ensure the path length in constructed
solution is within (1+$\epsilon$) times of shortest path in $G$, i.e.
$\cmstvcr(T,v_i) \leq (1+\epsilon)\cmstvcr(\cmstspath{v_i}, v_i)$, while total
storage cost is within (1+$\frac{2}{\epsilon}$) times of MST. In our problem,
the co-usage constraints of matrices in each snapshot form hyperedges over the
graph making the problem more difficult. 

In the rest of the discussion, we adapt meta-heuristics for constrained MST
problems to develop two algorithms: the first one (\weightstore-MT) is based on
an iterative refinement scheme, where we start from an MST and then adjust it to
satisfy constraints\eat{, similar to the LMT algorithm proposed in early work
\cite{vldb15versioning}}; the second one is a priority-based tree construction
algorithm (\weightstore-PT), which adds nodes one by one and encodes heuristic
in the priority function. Both algorithms aim to solve the parallel and
independent recreation schemes, and thus can also find feasible solution for reusable
scheme. Due to large memory footprints of intermediate matrices, we leave
improving reusable scheme solutions for future work.

\begin{algorithm}[t!] 
\caption{\weightstore-MT}
\label{alg:pas_mt}
{\small{
\begin{algorithmic}[1]
\REQUIRE $\cmstvg$, snapshots $S$, recreation cost $\{\theta_i \geq 0\ |\ s_i \in S\}$.
\ENSURE A spanning tree $T$ satisfying constraints $\{\cmstvcr(T, s_i) \leq \theta_i\}$
\STATE let $T = $ MST of $\cmstvgsimp$; 
\WHILE {unsatisfied constraints $U = \{s_i\ |\ \cmstvcr(T, s_i) > \theta_i\} \neq \emptyset$}
\FOR{each edge $e_{si} = (v_s, v_i) \in \cmstvgeset - T$}
  \STATE calculate $gain(e_{si})$  with Eq.~\ref{eq:marginal} (Eq.~\ref{eq:marginal_pa} for scheme $\cmstscheme{p}$)
\ENDFOR
\STATE find $e'_{si} = \max\{e_{si}\ |\ gain(e_{si}) \}$
\STATE \textbf{break} if $gain(e'_{si}) \leq 0$
\STATE \emph{swap} $(p_i, v_i)$ with $e'_{si}$: $T = (T - \{(p_i, v_i)\} ) \cup \{e'_{si}\}$
\ENDWHILE
\RETURN $T$ unless $U \neq \emptyset$
\end{algorithmic}
}}
\end{algorithm}

\topic{\weightstore-MT:} The algorithm starts with $T$ as the MST of $\cmstvgsimp$, and iteratively adjusts $T$ to satisfy the broken snapshot recreation 
constraints, $U = \{s_i | \cmstvcr(T, s_i) > \theta_i\}$, by swapping one edge at a time. 
We denote $p_i$ as the parent of $v_i$, $(p_i, v_i) \in T$ and $p_0 = \phi$, and successors of $v_i$ in $T$ as $\mathcal{D}_i$. A \underline{\emph{swap operation}} on $(p_i, v_i)$ to edge $(v_s, v_i) \in \cmstvgeset - T$ changes parent of $v_i$ to $v_s$ in $T$. 
\begin{lemma}
\label{lm:swap}
A swap operation on $v_i$ changes storage cost of $\cmstvcs(T)$ by $\cmstecs(p_i, v_i) - \cmstecs(v_s, v_i)$, and changes recreation costs of $v_i$ and its successors $ \mathcal{D}_i$ by: $\cmstvcr(T, v_i) - \cmstvcr(T, v_s) - \cmstecr(v_s, v_i)$. 
\end{lemma}

The proof can be derived from definition of $C_s$ and $C_r$ by inspection. When selecting edges in $\cmstvgeset - T$, we choose the one which has the largest marginal gain for unsatisfied constraints:

{\small{
\begin{align}
\label{eq:marginal}
\cmstscheme{i}:& \max_{(v_s, v_i) \in \cmstvgeset - T}
\{
\frac{\sum_{s_k \in U} \sum_{v_j \in s_k \cap \mathcal{D}_i}{(\cmstvcr(T, v_i) - \cmstvcr(T, v_s) - \cmstecr(v_s, v_i))}}{\cmstecs(v_s, v_i) - \cmstecs(p_i, v_i)}
\} & \\
\cmstscheme{p}:& \max_{(v_s, v_i) \in \cmstvgeset - T}
\{
\frac{\sum_{s_k \in U}{(\cmstvcr(T, v_i) - \cmstvcr(T, v_s) - \cmstecr(v_s, v_i))}}{\cmstecs(v_s, v_i) - \cmstecs(p_i, v_i)}
\} &
\label{eq:marginal_pa}
\end{align}
}}

The actual formula used is somewhat more complex, and handles negative denominators.
Eq.~\ref{eq:marginal} sums the gain of recreation cost changes among all matrices in the same snapshot $s_i$ (for the independent scheme), while Eq.~\ref{eq:marginal_pa} uses the max change instead (for the {\em parallel} scheme).

The algorithm iteratively swaps edges and stops if all recreation constraints are satisfied or no edge returns a positive gain. A single step examines $|\cmstvgeset\ - T|$ edges and $|U|$ unsatisfied constraints, and there are at most $|\cmstvgeset|$ steps. Thus the complexity is bounded by $O(|\cmstvgeset|^2|S|)$.

\begin{algorithm}[t!] 
\caption{\weightstore-PT}
\label{alg:pas_pt}
{\small{
\begin{algorithmic}[1]
\REQUIRE $\cmstvg$, snapshots $S$, recreation cost $\{\theta_i \geq 0\ |\ s_i \in S\}$.
\ENSURE A spanning tree $T$ satisfying constraints $\{\cmstvcr(T, s_i) \leq \theta_i\}$
\STATE let $T = \emptyset$ and $Q$ be a priority queue of edges based on $\cmstecs$
\STATE \textbf{push} $\{(\cmstroot, v_i)\ |\ v_i \in \cmstvgvset\} $ in $Q$
\WHILE{$Q \neq \emptyset$}
\STATE \textbf{pop} $e_{ij} = (v_i, v_j)$ from $Q$; let $T = T \cup \{e_{ij}\}$
\STATE let constraints satisfaction flag be $\Theta^{e_{ij}}_{satisfy} = true$
\FOR{each snapshot constraint $s_a \in \{s\ |\ s \in S \wedge v_j \in s\}$}
\STATE estimate recreation cost $\hat\cmstvcr(T, s_a)$
\STATE $\Theta^{e_{ij}}_{satisfy} = false$ and \textbf{break} if $\hat\cmstvcr(T, s_a) > \theta_a$
\ENDFOR
\STATE \textbf{if} $\Theta^{e_{ij}}_{satisfy}$ is $false$, \textbf{then} $T = T - \{e_{ij}\}$ and \textbf{goto} line 3
\STATE \textbf{pop} inner edges of $v_j$ $\mathcal{I}^j_T = \{(v_k, v_j)\ |\ v_k \in T\}$ from $Q$
\STATE \textbf{push} outer edges $\mathcal{O}^j_{\cmstvgeset - T} = \{(v_j, v_k)\ |\ v_k \in \cmstvgeset - T\}$ to $Q$
\FOR {$(v_k, v_j) \in T$, change $p_k$ improves $\cmstvcs$, and no worse $\cmstvcr$}
\STATE \emph{swap} $(p_k, v_k) \in T$ with $(v_j, v_k)$
\ENDFOR
\ENDWHILE
\IF{$T$ is not a spanning tree}
\STATE{\textbf{for each} $v_u \in \cmstvgvset - V_T$, \textbf{do} $T = T \cup \{e_{0u} = (\cmstroot, v_u)\}$}
\STATE \emph{adjust} $T$ using \weightstore-MT heuristic. 
\ENDIF
\RETURN $T$ if $T$ is a \emph{matrix storage plan}
\end{algorithmic}
}}
\end{algorithm}

\topic{\weightstore-PT:}
This algorithm constructs a solution by ``growing'' a tree starting with an empty tree.
The algorithm examines the edges in $\cmstvgsimp$ in the increasing order by the storage cost $\cmstecs$; a priority
queue is used to maintain all the candidate edges and is populated with all the edges from $v_0$ in the beginning. 
At any point, the edges in $Q$ are the ones that connect a vertex $T$, to a vertex outside $T$.
Using an edge $e_{ij} = (v_i, v_j)$ (s.t., $v_i \in V_T \wedge v_j \in \cmstvgvset-V_T$) popped from $Q$, the algorithm
tries to add $v_j$ to $T$ with minimum storage increment $\cmstecs(e_{ij})$.
Before adding $v_j$, it examines whether the constraints of affected groups $s_a$ (s.t., $v_j \in s_a$) are satisfied 
using actual and estimated recreation costs
for vertices $\{v_k \in s_a\}$ in $T$ and $\cmstvgvset-T$ respectively; if $v_k
\in T$, actual recreation cost $\cmstvcr(T, v_k)$ is used, otherwise the lower
bound of it, i.e. $\cmstecr(\cmstroot, v_k)$ is used as an estimation. We refer
the estimation for $s_a$ as $\hat\cmstvcr(T, s_a)$. 

Once an edge $e_{ij}$ is added to $T$, the inner edges $\mathcal{I}^j_T = \{(v_k, v_j) | v_k \in T\}$ of newly added $v_j$ are dequeued from $Q$, while the outer edges $\mathcal{O}^j_{\cmstvgeset - T} = \{(v_j, v_k)\ |\ v_k \in \cmstvgeset - T\}$ are enqueued. If the storage cost of existing vertices in $T$ can be improved (i.e., $\cmstvcs(T, v_k) > \cmstecs(v_k, v_j)$), and recreation cost is not more (i.e. $\cmstvcr(T, v_k) \geq \cmstvcr(T, v_j) + \cmstecr(v_k,v_j)$), then the parent $p_k$ of $v_k$ in T is replaced to $v_j$ via the swap operation, which obviously decreases the storage cost and affected group recreation cost. 

The algorithm stops if $Q$ is empty and $T$ is a spanning tree. In the case when $Q$ is empty but $V_T \subset \cmstvgvset$, an \underline{\emph{adjustment operation}} on $T$ to increase storage cost and satisfy the group recreation constraints is performed. 
For each $v_u \in \cmstvgvset-V_T$, we append it to $\cmstroot$, then in each unsatisfied group $s_i$ that $v_u$ belongs to, optimally, we want to choose a set of $\{v_g\} \subseteq s_i \cap T$ to change their parents in $T$, such that the decrement of storage cost is minimized while recreation cost is satisfied. The optimal adjustment itself can be viewed as a knapsack problem with extra non-cyclic constraint of $T$, which is NP-hard. Instead, we use the same heuristic in Eq.~\ref{eq:marginal} to adjust $v_g \in s_i \cap T$ one by one by replacing its parent $p_g$ to $v_s$ until the group constraint in $s_i$ is satisfied.
\eat{
{\small{
\begin{equation*}
\max_{v_s \in T \wedge v_s \neq p_g}
\{
\frac{\delta^{p_g \rightarrow v_s}_r}{\delta^{p_g \rightarrow v_s}_s} = 
\frac{(d_g+1)(\mathcal{C}_r(T, v_s) + c_r(v_s, v_g) - \mathcal{C}_r(T, v_g))}{c_s(v_s, v_g) - \mathcal{C}_s(T, v_g)}
\}
\end{equation*}
}}
where $d_g$ is number of successors of $v_g$ in $T$. 
}\eat{
In other words, we choose a neighbor $v_s$ to replace $p_g$, having the maximum marginal gain of recreation in the unsatisfied groups w.r.t. the storage increment.} As before, the parallel scheme $\cmstscheme{p}$ differs from independent case $\cmstscheme{i}$ in the adjustment operator using Eq.~\ref{eq:marginal_pa}. The complexity of this algorithm is $O(|\cmstvgeset|^2|S|)$.

\subsection{Model Evaluation Scheme in {\weightstore}}
\label{subsec:pas_evaluation}
Model evaluation, i.e., applying a \dnn\ forward on a data point to get the prediction result, is a common task to explore, debug and understand models. 
Given a \weightstore\ storage plan, an \dlvcmd{eval} query requires uncompressing and applying deltas along the path to the model. We develop 
a novel model evaluation scheme utilizing the segmented design, that progressively accesses the low-order segments only when necessary, and 
guarantees no errors for arbitrary data points.

The basic intuition is that: when retrieving segmented parameters, we know the minimum and maximum values of the parameters (since higher order bytes are retrieved first). 
If the prediction result is the same for the entire range of those values, then we do not need to access the lower order bytes. 
However, considering the high dimensions of parameters, non-linearity of the \dnn\ model, unknown full precision value when issuing the query, it is not clear if this is feasible.

We define the problem formally, and illustrate the determinism condition that we use to develop
our algorithm. 
Our technique is
inspired from theoretical stability analysis in numerical analysis. We make the
formulation  general to be applicable to other prediction functions. The
basic assumption is that the prediction function returns a vector showing relative
strengths of the classification labels, then the dimension index with the maximum
value is used as the predicted label. 

\begin{problem}[Parameter Perturbation Error Determination]
Given a prediction function $\mathcal{F}(d, W): \mathbb{R}^m \times \mathbb{R}^n \mapsto \mathbb{R}^c$, where $d$ is the data and $W$ are the learned weights, the prediction result $c_d$ is the dimension index with the highest value in the output $o \in \mathbb{R}^c$. When $W$ value is uncertain, i.e., each $w_i \in W$ in known to be in the range $[w_{i,\min }, w_{i, \max}]$, determine whether $c_d$ can be ascertained without error.
\end{problem}

When $W$ is uncertain, the output $o$ is uncertain as well. However, if we can bound the individual entries in $o$, then the following condition is an applicable necessary condition for determining error:

\begin{lemma}
\label{lm:error_determinism}
Let $o_i \in o$ vary in range $[o_{i,\min}, o_{i,\max}]$. If $\exists k$ such that $\forall i,\ o_{k,\min} > o_{i,\max}$, then prediction result $c_d$ is $k$.
\end{lemma}

Next we illustrate a query procedure, that given data $d$, evaluates a \dnn\ with weight perturbations and determines the output perturbation on the fly. Recall that \dnn\ is a nested function (Sec.~\ref{sec:preliminary}), \eat{where the input data is transformed layer by layer, and each layer is of the form:
\begin{displaymath}
\begin{aligned}
       & f_0 = \sigma_0(W_0 d + b_0) &\quad d \in \mathbb D \\
       & f_i = \sigma_i(W_i f_{i-1} + b_i) &\quad 0 < i \leq n \\
\end{aligned}
\end{displaymath}
}
we derive the output \eat{$o = f_n$ }perturbations when evaluating a model while preserving perturbations step by step:

{\small{
\begin{displaymath}
\begin{aligned}
x_{0,k} &= \sum_j W_{0,k,j} d_j + b_{0,k} & 
x_{0,k,\min} &= \sum_j \min\{W_{0,k,j} d_j\} + \min\{ b_{0,k} \} &\\
& & x_{0,k,\max} &= \sum_j \max\{W_{0,k,j} d_j\} + \max\{ b_{0,k} \} & \\
\end{aligned}
\eat{
x_{0,k,\min} = \sum_j \min\{W_{0,k,j} d_j\} + \min\{ b_{0,k} \}; 
x_{0,k,\max} = \sum_j \max\{W_{0,k,j} d_j\} + \max\{ b_{0,k} \} 
}
\end{displaymath}
}}

{\noindent}Next, activation function $\sigma_0$ is applied. Most of the common activation functions are monotonic function $\mathbb{R}\mapsto\mathbb{R}$, (e.g. sigmoid, ReLu), while pool layer functions are $\min$, $\max$, avg functions over several dimensions. It is easy to derive the perturbation of output of the activation function, $[f_{0,k,\min}, f_{0,k,\max}]$. During the evaluation query, instead of 1-D actual output, we carry 2-D perturbations, as the actual parameter value is not available. Nonlinearity decreases or increases the perturbation range. Now the output perturbation at $f_i$ can be calculated similarly, except now both $W$ and $f_{i-1}$ are uncertain:

{\small{
\begin{displaymath}
\begin{aligned}
x_{i,k} &= \sum_j W_{i,k,j} f_{i-1,j} + b_{i,k} & 
x_{i,k,\min} &= \sum_j \min\{W_{i,k,j} f_{i-1,j}\} + \min\{ b_{i,k} \} &\\
& & x_{i,k,\max} &= \sum_j \max\{W_{i,k,j} f_{i-1,j}\} + \max\{ b_{i,k} \} & \\
\end{aligned}
\eat{
x_{i,k,\min} = \sum_j \min\{W_{i,k,j} f_{i-1,j}\} + \min\{ b_{i,k} \}; 
x_{i,k,\max} = \sum_j \max\{W_{i,k,j} f_{i-1,j}\} + \max\{ b_{i,k} \} 
}
\end{displaymath}
}}

{\noindent}Applying these steps iteratively until last layer,  we can then apply Lemma~\ref{lm:error_determinism}, the condition of error determinism, to check if the result is correct. If not, then lower order segments of the float matrices are retrieved, and the evaluation is re-performed.

This progressive evaluation query techniques dramatically improve the utility of \weightstore, as we further illustrate in our experimental evaluation. 
Note that, other types of queries, e.g., matrix plots, activation plots, visualizations, etc., can often be executed without retrieving the lower-order bytes either.

\begin{figure*}[!t]
\subfigure[Compression-Accuracy Tradeoff\newline for Float Representation Schemes]{
  \includegraphics[totalheight=0.165\linewidth]{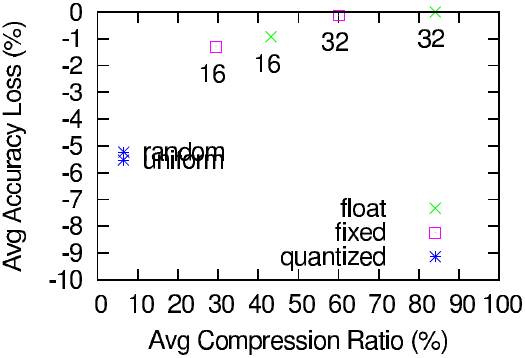} 
  \label{fig:exp1}
}\subfigure[Compression Performance for\newline Different Delta Schemes \& Models]{
  \includegraphics[totalheight=0.165\linewidth]{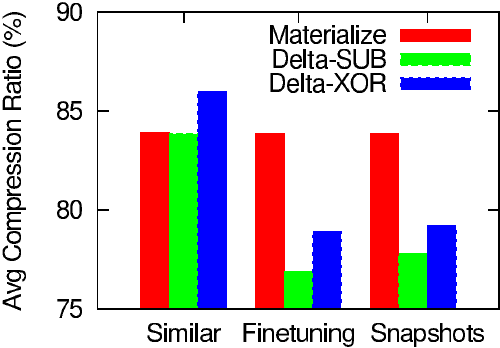} 
  \label{fig:exp2}
}\subfigure[\weightstore\ Optimal Archival Storage\newline Algorithms Results for \syntheticds ]{
  \includegraphics[totalheight=0.165\linewidth]{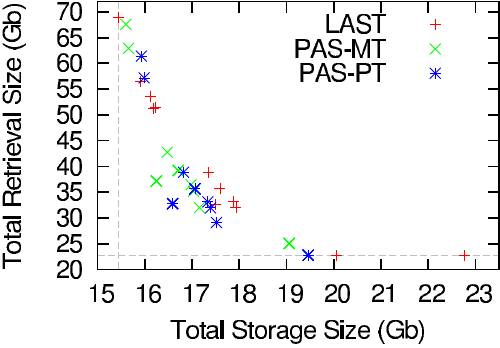} 
  \label{fig:exp3}
}\subfigure[Progressive Evaluation Query\newline Processing Using High-Order Bytes]{
  \includegraphics[totalheight=0.165\linewidth]{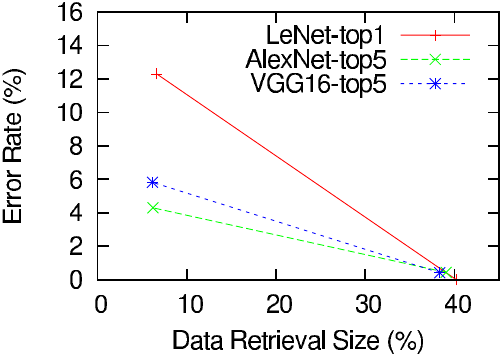}
  \label{fig:exp4}
}
\caption{Evaluation Results for \weightstore}
\label{fig:paseval}
\end{figure*}

\eat{
\begin{figure}[!t]
\centering
\includegraphics[width=0.9\linewidth]{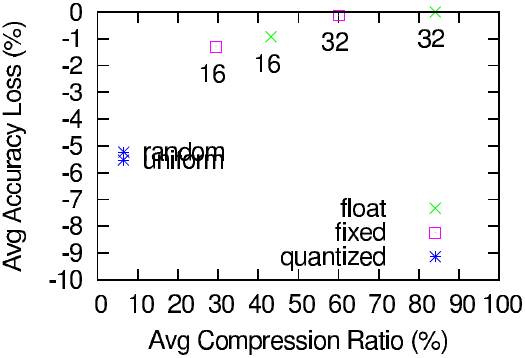}
\caption{Compression-Accuracy Tradeoff in Float Representation Schemes}
\label{fig:exp1}
\end{figure}
}

\section{Evaluation Study}
\label{sec:experiments}
\modelhub\ is designed to work with a variety of deep learning backends; our
current prototype interfaces with \cmd{caffe}~\cite{caffe2014mm} through a wrapper that can
extract caffe training logs, and read and write parameters for training. We have
also built a custom layer in \cmd{caffe} to support progressive queries. 
The \dlv\ command-line suite is implemented as a Ruby
gem, utilizing \cmd{git} as internal VCS and \cmd{sqlite3} and \weightstore\ as
backends to manage the set of heterogeneous artifacts in the local client.
\weightstore\ is built in C++ with gcc 5.4.0. All experiments are conducted on a
Ubuntu Linux 16.04 machine with an 8-core 3.0GHz AMD FX-380 processor, 16GB
memory,
and NVIDIA GTX 370 GPU. We use zlib for compression; unless
specifically mentioned, the compression level is set to 6. When wrapping
and modifying caffe, the code base version is rc3\eat{ \footnote{\cmd{caffe}:
{\small{\url{https://github.com/BVLC/caffe/releases/tag/rc3}}}}}.

In this section, we present a comprehensive evaluation with real-world and
synthetic datasets aimed at examining our design decisions, differences of
configurations in \weightstore, and performance of archiving and progressive
query evaluation techniques proposed in earlier sections. 

\subsection{Dataset Description}
\topic{Real World Dataset}: 
To study the performance of \weightstore\ design decisions, float
representations, segmented float storage, delta encoding and overall
compression performance, we use a collection of shared caffe models 
published in caffe repository or Model Zoo\eat{ (Table~\ref{tb:realmodel})}. In brief,
LeNet-5~\cite{lenet90} is a convolutional DNN with 431k parameters. The reference model has
0.88\% error rate on MNIST. AlexNet~\cite{alexnet2012imagenet} is a medium-sized model with 61 million
parameters, while VGG-16~\cite{vgg14} has 1.9 billion parameters. Both AlexNet and VGG-16
are tested on ILSVRC-2012 dataset. The downloaded models have 43.1\%, and
31.6\% top-1 error rate respectively. 
Besides, to study the delta performance on model repositories under different workloads (i.e., retraining, fine-tuning): we use CNN-S/M/F~\cite{vggsmfbmvc14}, a set of similar models developed by VGG authors to study model variations. These models are similar to VGG in architecture, and retrained from scratch; for fine-tuning, we use VGG-Salient~\cite{zhang2015salient} a fine-tuning VGG model which only changes last full layer. 

\eat{
\begin{table}[!h]
\centering
\begin{tabular}{cccc} 
\toprule
\textbf{Network}  & \textbf{$|W|$ (flops)} & \huicomment{}Size &  Purpose 
\\ 
\midrule
LeNet-5 & $4.31 \times 10^{5}$ & 5MB & Small model\\ 
AlexNet & $6\times10^7$ & 200MB & Medium model \\ 
VGG-16 \cite{vgg14} & $1.96 \times 10^{10}$ & 500MB & Large model \\ 
CNN-S/M/F \cite{vggsmfbmvc14} & $1.13\times10^{10}$ & 199/300MB & Similar models \\ 
VGG-Salient\huicomment{} & $1.96 \times 10^{10}$ & 500MB & Fine-tuning \\
\bottomrule
\end{tabular}
\caption{Real world \dnn\ models}
\label{tb:realmodel}
\end{table}
}

\topic{Synthetic Datasets:} 
Lacking sufficiently fine-grained real-world repositories of models, 
to evaluate \eat{optimality and scalability}performance of parameter archiving algorithms, we developed an automatic modeler \eat{using \DQL} to enumerate models and hyperparameters to produce a \dlv\ repository. \eat{We generated two synthetic datasets (SD1, SD2); SD1 mimics a modeler who is enumerating models to solve a face recognition task, while SD2 simulates a modeler fine-tuning a trained network for a new prediction task. SD1 exhibits higher variety in the network architectures and parameters, whereas the fine-tuning practice used in SD2 results in similar network architectures across the models and relatively similar parameters.}
We generated a synthetic dataset (\syntheticds): simulating a modeler who is enumerating models to solve a face recognition task, and fine-tuning a trained VGG. \syntheticds\ results in similar DNNs and relatively similar parameters across the models.
The datasets are shared online\footnote{Dataset Details: {\small{\url{http://www.cs.umd.edu/~hui/code/modelhub}}}}. 

To elaborate, the automation is driven by a state machine that applies modeling practices from the real world. \eat{For SD1, the modeler mutates the network architecture intensively by inserting/deleting layers and changing layer shapes, as well as by updating optimization related hyperparameters. 
In the second task, }The modeler updates the VGG network architecture slightly and changes VGG object recognition goal to a face prediction task (prediction labels changed from 1000 to 100, so the last layer is changed); \eat{in this scenario,} various fine-tuning hyperparameter alternations are applied by mimicking practice~\cite{fastrnn}. \eat{Details of these generation and \DQL\ query used are in Appendix.} \eat{Table~\ref{tb:syntheticmodels} summarizes the two datasets. } \syntheticds\ in total has 54 model versions, each of which have 10 snapshots. A snapshot has 16 parametric layers and a total of $1.96\times10^{10}$ floats.

\eat{
\begin{table}[!h]
{\small{
\centering
\begin{tabular}{ccccc} 
\toprule
\textbf{Network}  & \textbf{Versions} & \textbf{Snapshots} & \textbf{Matrices} & \textbf{Type}
\\ 
\midrule
SD1  & 50 & 10 x 50 & [a, b] & Retraining\\ 
\midrule
SD2  & 54 & 10 x 54 & 16 & Fine-tuning\\ 
\bottomrule
\end{tabular}
}}
\caption{Synthetic \dnn\ Model Versions}
\label{tb:syntheticmodels}
\end{table}

Note SD1 is smaller than SD2, as SD1 requires end-to-end training which takes weeks on limited hardware resource. We use a medium size network but smaller images to lowering required parameters. 
}

\subsection{Evaluation Results}

\topic{Float Representation \& Accuracy:}
We show the effect of different float encoding schemes on the compression and accuracy in Fig.~\ref{fig:exp1}; this is a tradeoff that the user often needs to consider when configuring \modelhub\ to save a model. In Fig.~\ref{fig:exp1}, for each scheme, we  plot the average compression ratio versus the average accuracy drop when applying \weightstore\ float schemes on the three real world models. Here, {\em random} and {\em uniform} denote two standard quantization schemes. As we can see, we can get very high compression ratios (a factor of 20 or so) without a significant loss in accuracy, which may be acceptable in many scenarios.


\topic{Delta Encoding \& Compression Ratio Gain:}
Next we study the usefulness of delta encoding in real world models in the following scenarios: \textbf{a)} \emph{Similar}: latest snapshots across similar models (CNN-S/M/F, VGG-16); \textbf{b)} \emph{Fine-tuning}: fine-tuning models (VGG-16, VGG-Salient); and \textbf{c)} \emph{Snapshots}: snapshots for the same VGG models in \syntheticds\ between iterations. 
In Fig.~\ref{fig:exp2}, for different delta schemes, namely, storing original matrices (\emph{Materialize}), arithmetic subtraction (\emph{Delta-SUB}), and bitwise XOR diff (\emph{Delta-XOR}), the comparison is shown (i.e., we show the results of compressing the resulting matrices using {\bf zlib}). 
The figure shows the numbers under lossless compression scheme (float 32), which has 
the largest storage footprint. 

As we can see, delta scheme is not always good, due to the non-convexity and high entropy of parameters. For models under similar architectures, storing materialized original parameters is often better than applying delta encoding. With fine-tuning and nearby snapshots, the delta is always better, and arithmetic subtraction is consistently better than bitwise XOR. We saw similar results for many other models. These findings are useful for \weightstore\ implementation decisions, where we only perform delta between nearby snapshots in a single model, or for the fine-tuning setting among different models.

\eat{
\begin{figure}[!h]
\centering
\includegraphics[width=0.85\linewidth]{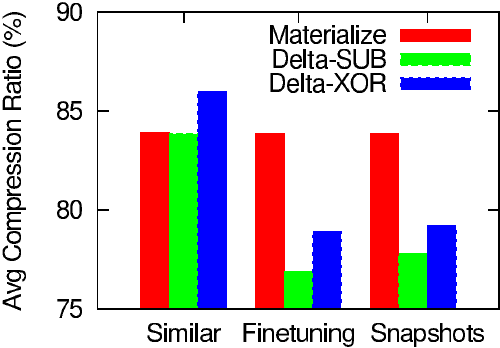}
\caption{Compression Performance for Different Delta Schemes \& Models}
\label{fig:exp2}
\end{figure}
}

\eat{
\begin{figure}[!h]
\centering
\includegraphics[width=0.85\linewidth]{mst_edge_stat_compact.eps}
\caption{Edge Relative Weight Distribution in the Version Graph}
\label{fig:exp3_data}
\end{figure}
}

\topic{Optimal Parameter Archival Storage:}
We study the optimality and scalability of \weightstore-PT and \weightstore-MT with the baseline LAST~\cite{khuller1995balancing} on the optimal parameter archival problem. We use \syntheticds\ here, from which we derive nearby snapshot deltas as well as model-wise deltas among the latest snapshots. \eat{Then we sample the model version graph to get various sizes of repository as the instance of the problem.} For the access frequencies, in practice, the modeler tends to access better models more frequently than bad ones. We use \eat{Zipf distribution as well as} multinomial distributions based on model accuracies to mimic the behavior. \eat{The delta edge weight distribution for \eat{SD1 and} SD2 is shown in Fig.~\ref{fig:exp3_data}.}
As LAST cannot handle the co-usage constraints, to have a reasonable baseline, we decompose the co-usage constraint to be proportional to each layer matrix size. The performance of the algorithms 
is shown in Fig.~\ref{fig:exp3}. Two dotted lines are the minimum possible storage cost (MST) and the recreation cost (SPT) for the independent scheme in version graph. Each data point is a storage plan found by the algorithms given a recreation cost constraint. As we can see, \weightstore\ algorithms (PT and MT) consistently perform better than LAST. The reason is both of them consider the constraints as a group and thus choose better storage plans.

\eat{
\begin{figure}[!h]
\centering
\includegraphics[width=0.85\linewidth]{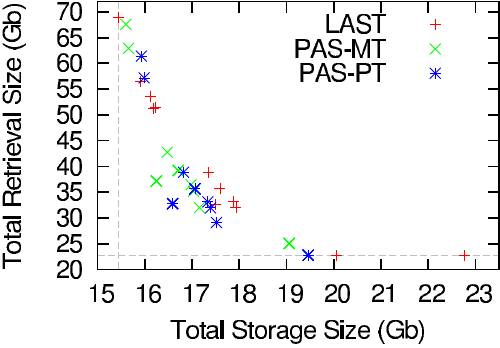}
\caption{\weightstore\ Version Graph Compression Results in SD2}
\label{fig:exp3}
\end{figure}
}

\eat{
\huicomment{y-axis to use weighted recreation needs }
\huicomment{extra plot for parallel scheme..}
\huicomment{scalability is omitted, as the SD2 is relatively small. runtime complexity in the algorithm sections can serve the purpose.}
}

\eat{
Next we show the scalability of \weightstore-MT and \weightstore-PT by vary the number of repositories. As in practice, the repository of handcrafted by modelers, we didn't show very large instances.
}

\topic{Progressive Query Evaluation:}
We study the efficiency of the progressive evaluation technique using perturbation error determination scheme on real world models (LeNet, AlexNet, VGG16) and their corresponding datasets. The original parameters are 4-byte floats, which are archived in segments in \weightstore. We modify \cmd{caffe} implementation of involved layers and pass two additional blobs (min/max errors) between layers. The perturbation error determination algorithm uses high order segments, and answers \cmd{eval} query on the whole test dataset. The algorithm determines whether top-k (1 or 5) result needs lower order bytes (i.e., matched index value range overlaps with $k+1$ index value range). 

The result is summarized in Fig.~\ref{fig:exp4}. The y-axis shows the error rate, i.e., the percentage of test dataset that may have perturbation errors due to low precision. The x-axis shows the percentage of data that needs to be retrieved (i.e., 2 bytes or 1 byte per float). As one can see, the prediction errors requiring full precision lower-order bytes are very small. The less high-order bytes used, higher the chance of potential errors. The consistent result of progressive query evaluation on real models supports our design decision of segmented float storage. 

\eat{
\begin{figure}[!h]
\centering
\includegraphics[width=0.85\linewidth]{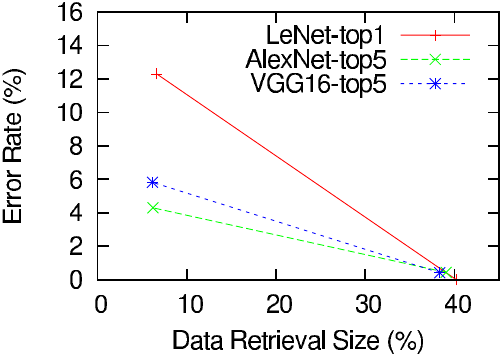}
\caption{Progressive Evaluation Query Processing Using High-Order Chunks}
\label{fig:exp4}
\end{figure}
}

\eat{
\huicomment{more data points is better? e.g. retrieve bits instead of 2byte/1byte?}
}

\eat{
\topic{End-to-End \weightstore\ Performance}:
To show the end-to-end performance of \weightstore, we report the total compression that \weightstore\ can archive on SD1 and SD2. Using segmentation, \weightstore\ can support queries at most cases without loading the lower order segmented, which can be saved on remote places or physical archival disks.


\huicomment{end-to-end is not conducted. need it?}

\huicomment{compress modelzoo models?}
}

\section{Related Work}
\label{sec:related_work}

\topic{Machine Learning Systems:} 
There have been several high-profile deep learning systems in recent years, but those typically focus on the training aspects (e.g., on distributed training, how to utilize GPUs or allow symbolic formulas, etc.)~\cite{caffe2014mm,tensorflow,distbelief2012nips,adam14osdi,singa2015mm}. 
The data \eat{management }and lifecycle management challenges discussed above have been largely ignored so far, but are becoming critical as the use of deep learning permeates through a variety of application domains, since those pose a high barrier to entry for many potential users.
In the database community, there has been increasing work on developing general-purpose systems for supporting machine learning\eat{~\cite{arun2015sigmod,mli2013spark,graphlab2012vldb}}, including pushing predictive models into databases~\cite{brown11cidr,madlib2012sigmod},
   accelerating tasks using database optimizing methods\eat{ physical design}~\cite{ce2014sigmod,arun2015sigmod}, 
and managing modeling lifecycles and serving predictive models in advanced ways~\cite{madden15learningsys,franklin2015learningsys}.  
\modelhub\ is motivated by similar principles; aside from a focus on DNNs, it also supports {\em versioning} as a first-class construct~\cite{datahub} which differentiates it from that work.

\topic{DNN Compression:} 
There has been increasing interest on compressing \dnn\ models, motivated in
part by the need to deploy them on devices with simple instruction sets, low
memory, and/or energy constraints~\cite{han2015deep,denton2014exploiting,sung2015resiliency}. 
However, the goal of those works is simplify the model in a lossy manner with as
little loss of accuracy as possible, which makes this work orthogonal to
\eat{lossless compression}the archival approach we take in \modelhub; in fact, simplified models
are likely to compress much better, magnifying the gains of our approach as
our experimental results show.
Further, these methods often require heavy retraining or expensive computations
({\em k}-means, SVD, etc.) to derive simpler models, which makes them too
heavyweight in an interactive setting for which \DLV\ is designed. 

\topic{DNNs with Low Precision Floats:} 
Low precision floats are exploited in accelerating training and testing systems \cite{vanhoucke2011improving,sung2015resiliency,gupta2015deep}\eat{courbariaux2014training}, showing techniques and empirical results when training, testing \dnn\ with limited precisions. \modelhub\ differs from their work by exploiting \eat{storing large collection of }parameters archiving, and use segmented floats to answer lifecycle modeling queries.

\topic{Stability Analysis Results} 
Stability analysis of \dnn\ is studied in the literature \cite{stevenson1990sensitivity,zeng2001sensitivity,yang2013computation}, where the problem setting is perturbation analysis regardless specific data, rather focus on statistical measures of stability. \modelhub\ uses basic perturbation analysis techniques and focus on novel progressive query answering in a segmented float storage.

\section{Conclusion and Future Work}
\label{sec:concolusion}
In this paper, we described some of the key data management challenges in learning, managing, and adjusting 
deep learning models, and presented our \modelhub\ system that attempts to address those challenges in a systematic
fashion. The goals of \modelhub\ are multi-fold: (a) to make it easy for a user to explore the space of 
potential models by tweaking the network architecture and/or the hyperparameter values, (b) to minimize the burden
in keeping track of the metadata including the accuracy scores and the fine-grained results, and (c) to compactly store
a large number of models and constituent snapshots without compromising on query or retrieval performance. 
We presented several high-level abstractions, including a command-line version management tool and a domain-specific
language, for addressing the first two goals. Anecdotal experience with our early users suggests that both of those are 
effective at simplifying the model exploration tasks. We also developed a read-optimized parameter archival storage for 
storing the learned weight parameters, and designed novel algorithms for storage optimization and for progressive query 
evaluation. Extensive experiments on real world and synthetic models verify the design decisions we made and demonstrate 
the advantages of proposed techniques.

\eat{
In this paper, we focused on describing and addressing

In this paper, we presented the \dnn\ modeling lifecycle and raise database
research challenges. The end-to-end learning method of \dnn\ expose modelers a
tunable black box. Where heuristic driven model enumeration via network
architecture adjustment and hyperparameter tuning sit in the core of the
lifecycle. A rich set of modeling artifacts (structure hyperparameter, graph
structure network architecture, binary parameters) are produced during the
process. We propose \modelhub, a client-server version control system to help
modelers manage modeling artifacts and ease their tasks. In \modelhub, we
abstract the lifecycle and propose declarative constructs, \dlv\ specialized
VCS interface, as well as \DQL\ domain specific model enumeration language. To
help modeler understand models, we propose model understanding schemes for both
network architecture and learned parameters. We further explore core database
management issues to manage the versioned float number parameters. A novel
chunked bit-wise float storage, \weightstore, is proposed. We show the
utilities of artifacts and study a new form of constraint in dataset versioning
setups. We also present a progressive query answering technique to benefit from
the segmented float stores. 
}

\balance

\eat{
{\small
\bibliographystyle{abbrv}
\bibliography{main}
}
}

\bibliographystyle{IEEEtran}
\def\IEEEbibitemsep{1pt}
\bibliography{IEEEabrv,main}

\begin{thebibliography}{10}
\providecommand{\url}[1]{#1}
\csname url@samestyle\endcsname
\providecommand{\newblock}{\relax}
\providecommand{\bibinfo}[2]{#2}
\providecommand{\BIBentrySTDinterwordspacing}{\spaceskip=0pt\relax}
\providecommand{\BIBentryALTinterwordstretchfactor}{4}
\providecommand{\BIBentryALTinterwordspacing}{\spaceskip=\fontdimen2\font plus
\BIBentryALTinterwordstretchfactor\fontdimen3\font minus
  \fontdimen4\font\relax}
\providecommand{\BIBforeignlanguage}[2]{{%
\expandafter\ifx\csname l@#1\endcsname\relax
\typeout{** WARNING: IEEEtran.bst: No hyphenation pattern has been}%
\typeout{** loaded for the language `#1'. Using the pattern for}%
\typeout{** the default language instead.}%
\else
\language=\csname l@#1\endcsname
\fi
#2}}
\providecommand{\BIBdecl}{\relax}
\BIBdecl

\bibitem{lecun2015nature}
Y.~LeCun, Y.~Bengio, and G.~Hinton, ``Deep learning,'' \emph{Nature}, 2015.

\bibitem{ce2014sigmod}
C.~Zhang, A.~Kumar, and C.~R{\'e}, ``Materialization optimizations for feature
  selection workloads,'' in \emph{SIGMOD}, 2014.

\bibitem{caffe2014mm}
Y.~Jia \emph{et~al.}, ``Caffe: Convolutional architecture for fast feature
  embedding,'' in \emph{ACM MM}, 2014.

\bibitem{tensorflow}
M.~Abadi \emph{et~al.}, ``{TensorFlow}: A system for large-scale machine
  learning,'' in \emph{OSDI}, 2016.

\bibitem{sgdtrick2012dnntrickbook}
L.~Bottou, ``Stochastic gradient descent tricks,'' in \emph{Neural Networks:
  Tricks of the Trade}.\hskip 1em plus 0.5em minus 0.4em\relax Springer, 2012,
  pp. 421--436.

\bibitem{lenet90}
Y.~LeCun \emph{et~al.}, ``Handwritten digit recognition with a back-propagation
  network,'' in \emph{NIPS}, 1990.

\bibitem{alexnet2012imagenet}
A.~Krizhevsky, I.~Sutskever, and G.~E. Hinton, ``Imagenet classification with
  deep convolutional neural networks,'' in \emph{NIPS}, 2012.

\bibitem{vgg14}
K.~Simonyan and A.~Zisserman, ``Very deep convolutional networks for
  large-scale image recognition,'' \emph{CoRR}, vol. abs/1409.1556, 2014.

\bibitem{resnet}
K.~He, X.~Zhang, S.~Ren, and J.~Sun, ``Deep residual learning for image
  recognition,'' in \emph{CVPR}, 2016.

\bibitem{datahub}
A.~Bhardwaj \emph{et~al.}, ``{DataHub}: Collaborative data science and dataset
  version management at scale,'' in \emph{CIDR}, 2015.

\bibitem{vldb15versioning}
S.~Bhattacherjee, A.~Chavan, S.~Huang, A.~Deshpande, and A.~Parameswaran,
  ``Principles of dataset versioning: Exploring the recreation/storage
  tradeoff,'' \emph{PVLDB}, 2015.

\bibitem{vanhoucke2011improving}
V.~Vanhoucke, A.~Senior, and M.~Z. Mao, ``Improving the speed of neural
  networks on {CPUs},'' in \emph{Proc. Deep Learning and Unsupervised Feature
  Learning NIPS Workshop}, 2011.

\bibitem{han2015deep}
S.~Han \emph{et~al.}, ``Deep compression: Compressing deep neural networks with
  pruning, trained quantization and huffman coding,'' in \emph{ICLR}, 2016.

\bibitem{courbariaux2014training}
M.~Courbariaux \emph{et~al.}, ``Training deep neural networks with low
  precision multiplications,'' \emph{arXiv preprint arXiv:1412.7024}, 2014.

\bibitem{schendel2012isobar}
E.~R. Schendel \emph{et~al.}, ``Isobar preconditioner for effective and
  high-throughput lossless data compression,'' in \emph{ICDE}, 2012.

\bibitem{bhattacherjee2014pstore}
S.~Bhattacherjee, A.~Deshpande, and A.~Sussman, ``Pstore: an efficient storage
  framework for managing scientific data,'' in \emph{SSDBM}, 2014.

\bibitem{deo1997computation}
N.~Deo and N.~Kumar, ``Computation of constrained spanning trees: A unified
  approach,'' in \emph{Network Optimization}, 1997.

\bibitem{vlsi1994optimal}
A.~B. Kahng and G.~Robins, \emph{On optimal interconnections for VLSI}.\hskip
  1em plus 0.5em minus 0.4em\relax Springer Science \& Business Media, 1994,
  vol. 301.

\bibitem{khuller1995balancing}
S.~Khuller, B.~Raghavachari, and N.~Young, ``Balancing minimum spanning trees
  and shortest-path trees,'' \emph{Algorithmica}, 1995.

\bibitem{vggsmfbmvc14}
K.~Chatfield \emph{et~al.}, ``Return of the devil in the details: Delving deep
  into convolutional nets,'' in \emph{BMVC}, 2014.

\bibitem{zhang2015salient}
J.~Zhang, S.~Ma, M.~Sameki, S.~Sclaroff, M.~Betke, Z.~Lin, X.~Shen, B.~Price,
  and R.~Mech, ``Salient object subitizing,'' in \emph{CVPR}, 2015.

\bibitem{fastrnn}
R.~Girshick, ``Fast {R-CNN},'' in \emph{ICCV}, 2015.

\bibitem{distbelief2012nips}
J.~Dean \emph{et~al.}, ``Large scale distributed deep networks,'' in
  \emph{NIPS}, 2012.

\bibitem{adam14osdi}
T.~Chilimbi \emph{et~al.}, ``Project {ADAM}: Building an efficient and scalable
  deep learning training system,'' in \emph{OSDI}, 2014.

\bibitem{singa2015mm}
W.~Wang \emph{et~al.}, ``Singa: Putting deep learning in the hands of
  multimedia users,'' in \emph{ACM MM}, 2015.

\bibitem{brown11cidr}
M.~Akdere, Cetintemel \emph{et~al.}, ``The case for predictive database
  systems: Opportunities and challenges.'' in \emph{CIDR}, 2011.

\bibitem{madlib2012sigmod}
X.~Feng, A.~Kumar, B.~Recht, and C.~R{\'e}, ``Towards a unified architecture
  for in-rdbms analytics,'' in \emph{SIGMOD}, 2012.

\bibitem{arun2015sigmod}
A.~Kumar, J.~Naughton, and J.~M. Patel, ``Learning generalized linear models
  over normalized data,'' in \emph{SIGMOD}, 2015.

\bibitem{madden15learningsys}
M.~Vartak \emph{et~al.}, ``Supporting fast iteration in model building,'' in
  \emph{LearningSys}, 2015.

\bibitem{franklin2015learningsys}
D.~Crankshaw, X.~Wang, J.~E. Gonzalez, and M.~J. Franklin, ``Scalable training
  and serving of personalized models,'' in \emph{LearningSys}, 2015.

\bibitem{denton2014exploiting}
E.~L. Denton \emph{et~al.}, ``Exploiting linear structure within convolutional
  networks for efficient evaluation,'' in \emph{NIPS}, 2014.

\bibitem{sung2015resiliency}
W.~Sung, S.~Shin, and K.~Hwang, ``Resiliency of deep neural networks under
  quantization,'' \emph{arXiv preprint arXiv:1511.06488}, 2015.

\bibitem{gupta2015deep}
S.~Gupta, A.~Agrawal, K.~Gopalakrishnan, and P.~Narayanan, ``Deep learning with
  limited numerical precision,'' in \emph{ICML}, 2015.

\bibitem{stevenson1990sensitivity}
M.~Stevenson \emph{et~al.}, ``Sensitivity of feedforward neural networks to
  weight errors,'' \emph{IEEE Trans. Neural Networks}, vol.~1, 1990.

\bibitem{zeng2001sensitivity}
X.~Zeng \emph{et~al.}, ``Sensitivity analysis of multilayer perceptron to input
  and weight perturbations,'' \emph{IEEE Trans. Neural Networks}, vol.~12,
  2001.

\bibitem{yang2013computation}
J.~Yang, X.~Zeng, and S.~Zhong, ``Computation of multilayer perceptron
  sensitivity to input perturbation,'' \emph{Neurocomputing}, 2013.

\end{thebibliography}

\end{document}